\documentclass[12pt,a4paper,reqno,intlimits,sumlimits]{amsart}

\usepackage{amsthm,amsmath, mathtools}
\usepackage{enumitem}
\usepackage[numbers,round]{natbib}
\usepackage{amssymb,dsfont,url}
\theoremstyle{plain}
\newtheorem{theorem}{Theorem}[section]
\newtheorem{lemma}[theorem]{Lemma}
\newtheorem{proposition}[theorem]{Proposition}
\newtheorem{corollary}[theorem]{Corollary}
\newtheorem{defin}[theorem]{Definition}
\newtheorem{remark}[theorem]{Remark}
\newtheorem{example}[theorem]{Example}

\newtheorem{conditions}[theorem]{Conditions}

\DeclareMathOperator{\loc}{loc}

\newcommand{\dd}[1]{\operatorname{d}\!#1}
\newcommand{\ccd}{\mathds{C}^d}
\newcommand{\ee}[1]{\operatorname{e}^{#1}}

\newcommand{\OF}{\mathcal{F}}
\newcommand{\OB}{\mathcal{B}}
\newcommand{\OP}{\mathcal{P}}
\newcommand{\notiz}[1]{\relax}

\newcommand{\zitep}[1]{\relax}
\newcommand{\skr}{\rangle}
\newcommand{\1}{\mathds 1}            

\newcommand{\nn}{\mathds N}
\newcommand{\N}{\mathds N}
 
\newcommand{\OA}{\mathcal{A\,}}

\newcommand{\OM}{\mathcal{M}}

\newcommand{\rr}{\mathds R}

\newcommand{\rrd}{\mathds{R}^d}

\newcommand{\cc}{\mathds C}
\newcommand{\skl}{\langle}
\newcommand{\sign}{\operatorname{sgn}}

\newcommand{\q}{\1_{\overline{D}^c}}

\newcommand{\Ukri}{\overset{\circ}{U}\vphantom{l}}

\newcommand{\tild}{~}
\newcommand{\trans}{\top}

\begin{document}

\title[Feynman-Kac formula for L\'evy processes]{Feynman-Kac formula for L\'evy processes \\with discontinuous killing rate}

\author[K. Glau]{Kathrin Glau}



\date{\today\\\indent Technische Universit{\"a}t M{\"u}nchen, Center for Mathematics\\ \indent kathrin.glau@tum.de}
\maketitle

\begin{abstract}
The challenge to fruitfully merge state-of-the-art techniques from  mathematical finance and numerical analysis has inspired researchers to develop fast deterministic option pricing methods. As a result, highly efficient algorithms to compute option prices in L\'evy models by solving partial integro differential equations have been developed. 
In order to provide a solid mathematical foundation for these methods, we derive a Feynman-Kac representation of variational solutions to partial integro differential equations that characterize conditional expectations of functionals of killed time-inhomogeneous L\'evy processes.
 We allow for a wide range of underlying stochastic processes, comprising processes with Brownian part as well as a broad class of pure jump processes such as generalized hyperbolic, multivariate normal inverse Gaussian, tempered stable, and $\alpha$-semi stable L\'evy processes. By virtue of our mild regularity assumptions as to the killing rate and the initial condition of the partial integro differential equation, our results provide a rigorous basis for numerous applications, 
 in financial mathematics, probability theory and physics.
 We reencounter the original ideas of Feynman and Kac, but now revealing the normal inverse Gaussian process in its role connecting the relativistic Schr\"odinger equation to stochastic processes. In Regard to finance we suggest a flexible class of employee options.
We implement a Galerkin scheme to solve the attendant pricing equation numerically and illustrate the effect of a killing rate.
\end{abstract}

\footnotesize{}

\section{Introduction}
%

Feynman-Kac formulas play a distinguished role in probability theory and functional analysis. Ever since their birth in 1949, Feynman-Kac-type formulas have been a constant source of fascinating insights in a wide range of disciplines. They originate in the description of particle diffusion by connecting Schr\"odinger's equation and the heat equation to the Brownian motion, see \cite{Kac1949}. A type of Feynman-Kac formula also figures at the beginning of modern mathematical finance: In their seminal article of 1973\nocite{black.scholes73}, Black and \text{Scholes} derived their Nobel Prize-winning option pricing formula by expressing the price as a solution to a partial differential equation, thereby rediscovering Feynman and Kac's deep link. 

  The fundamental contribution of Feynman-Kac formulas is to link stochastic processes to solutions of deterministic partial differential equations.  
Thus they also establish a connection between probability theory and numerical analysis, two disciplines that have evolved largely  separately.  
   Although both enjoy great success, transfer between them has remained only incidental.  This may very well be the reason for applications of Feynman-Kac still appearing so surprisingly fresh.
 In computational finance, they enable the development of option pricing methods by solving deterministic evolution equations. 
These have proven to be highly efficient, particularly when compared 
to Monte Carlo simulation. 
Thus, like other deterministic methods, they come into play whenever efficiency is essential and the complexity of the pricing problem is not too high. 
This is the case for recurring tasks, such as calibration and real-time pricing, and over the last few decades has given rise to extensive research in computing option prices by solving partial differential equations.
The challenge to extend these methods to price options in
advanced jump models has furthermore inspired researchers in recent years to develop highly efficient and widely applicable algorithms, see for instance
\cite{ContVoltchkova2005a}, Hilber, Reich, Schwab and Winter (2009)\nocite{HilberReichSchwabWinter2009}, Hilber, Reichmann, Schwab and Winter (2013)\nocite{HilberReichmannSchwabWinter2013}, Salmi, Toivanen and Sydow (2014)\nocite{SalmiToivanenSydow2014} and \cite{Itkin2015}.


%
%
%

In this article we derive a Feynman-Kac-type formula so as to provide a solid mathematical basis for fast option pricing in time-inhomogeneous L\'evy models using partial integro differential equations (PIDEs). While large parts of the literature 
focus on numerical aspects of these pricing methods, only little is known about the precise link between the related deterministic equations and the corresponding conditional expectations representing option prices.
 Our main question therefore is: 
 \emph{Under which conditions is there a Feynman-Kac formula linking option prices given by conditional expectations with solutions to evolution equations?}

 In order to further specify the problem, we focus on time-inhomogeneous L\'evy models and options whose path dependency may be expressed by a 
 killing rate. 
In this setting with $\OA=(\OA_t)_{[0,T]}$ the \emph{Kolmogorov operator} of a time-inhomogeneous L\'evy process, \emph{killing rate} (or \emph{potential}) $\kappa:[0,T]\times\rrd\to\rr$,  \emph{source} $f:[0,T]\times\rrd\to\rr$ and \textit{initial condition} $g:\rrd\to\rr$, the \emph{Kolmogorov equation} is of the form
\begin{align}\label{parabolic-eq-origin}
\begin{split}
\partial_t u + \OA_{T-t} u + \kappa_{T-t} u  &=\, f, \\
u(0) &=\, g\,.
\end{split}
\end{align}

Adopting a heuristic approach, one would typically assume that equation~\eqref{parabolic-eq-origin} has a classical solution $u$. 
If this solution is sufficiently regular to allow for an application of It\^o's formula and moreover satisfies an appropriate integrability condition, 
 the following Feynman-Kac-type representation 
\begin{equation}\label{gl-stochdarst-origin}
\scalebox{.94}[1]{$\displaystyle  u(T-t,L_t)
=
E\Big(g(L_T)\ee{-\int_t^T  \! \kappa_h(L_{h}) \dd h}
 + 
\int_t^{T}\!\!\! f(T\!-\!s,L_s)\ee{-\int_t^s  \! \kappa_h(L_{h}) \dd h} \dd s \,\Big|\,\OF_t\Big)$}
\end{equation}
follows by standard arguments 
and taking conditional expectations, see equations \eqref{kern-argument-Ito} and \eqref{Ewn} on page \pageref{Ewn} for a detailed derivation.
 Then, the conditional expectation \eqref{gl-stochdarst-origin} can be obtained by solving Kolmogorov equation \eqref{parabolic-eq-origin} by means of a deterministic numerical scheme. Such an argumentation hinges on a strong regularity assumption on the solution $u$ and thus implicitly on the data of the equation, $g$, $f$, $\OA$ and $\kappa$. We have to realize, however, that this constitutes  a serious restriction on the applicability of such a heuristic approach.

To do justice to 
the complexities of financial applications, we 
pay special attention to identifying appropriate conditions for the validity of equation~\eqref{gl-stochdarst-origin} for financial applications. Often, discontinuous killing rates constitute a natural choice, as we will show in several detailed examples in Section~\ref{sec-applications}. In particular indicator functions as killing rates turn out to be key to a wide variety of applications, both in mathematical finance and in probability theory. As one typical application we propose and study a flexible family of employee options in Section~\ref{sec-empop} and illustrate the numerical effect of such killing rates in Section~\ref{sec-num}. The fundamental role of killing rates of indicator type is killing the process outside a specified domain, which makes them attractive for applications. Moreover, they are closely related to occupation times and exit times of stochastic processes as we outline in Sections~\ref{sec-oc} and~\ref{sec-pendomain}. 
We furthermore find that discontinuous killing rates form a common root of exit probabilities of stochastic processes and  the distribution of supremum processes. As such they apply to the prices of path-dependent options like those of barrier, lookback, and American type.
In view of these considerations, which are both of a theoretical and applied nature, we will also want to allow for \emph{non-smooth} and even \emph{discontinuous} killing rates in Kolmogorov equation~\eqref{parabolic-eq-origin}.

Discontinuities in the killing rate $\kappa$ result in non-smoothness of the solution\tild $u$ of Kolmogorov equation \eqref{parabolic-eq-origin}. In particular, one cannot expect $u\in C^{1,2}$.
Assume $u(0)\neq0$ and $\kappa=\1_{(-\infty,0)^d}$ in \eqref{parabolic-eq-origin}, then $x\mapsto u(t,x) \in C^2$ implies $x\mapsto \kappa(x)u(t,x) \in C$, which obviously is a contradiction. Hence, for our purposes, the assumption that It\^o's formula can be applied to the solution $u$ is futile. Neither is it reasonable to assume that equation\tild \eqref{parabolic-eq-origin} has a classical solution. Let us emphasize that such irregularity is not only inherent in equation\tild \eqref{parabolic-eq-origin} if the killing rate is discontinuous, but also a typical feature of Kolmogorov equations for other path-dependent option prices. Prominent examples are boundary value problems related to barrier options in L\'evy models as well as free boundary value problems for American option prices.
In each of these cases, the use of a generalized solution concept is called for.

Among the possible generalizations of classical solutions of partial differential equations, we find that viscosity and weak solutions are the ones that are most commonly discussed.  Viscosity solutions directly abstract from pointwise solutions by introducing  comparison functions that are sufficiently regular, while the root of weak solutions is the problem formulation in a Hilbert space. Conceptually, both have their advantages. From a numerical perspective, viscosity solutions relate to finite difference schemes, 
whereas weak solutions are the theoretical foundation of Galerkin methods, a rich class of versatile numerical methods to solve partial differential equations. Relying on their elegant Hilbert space formulation, Galerkin methods by their very construction lead to convergent schemes as well as to a lucid error analysis. They furthermore distinguish themselves by their enormous flexibility towards problem types as well as compression techniques. Both theory and implementation of Galerkin methods have experienced a tremendous advancement over the past fifty years. They have become indispensable for today's technological developments in such diverse areas as aeronautical, biomechanical, and automotive engineering. 


In mathematical finance, Galerkin pricing algorithms have been developed for various applications, even for basket options in jump models.  Furthermore, numerical experiments and error estimates have confirmed their efficiency both in theory as well as in practice.
See \cite{HilberReichmannSchwabWinter2013}, and e.g. Matache, von Petersdorff and Schwab (2004)\nocite{MatachePetersdorffSchwab2004}, Matache, Schwab and Wihler (2005)\nocite{MatacheSchwabWihler2005}, von Petersdorff and Schwab (2004)\nocite{PetersdorffSchwab2004}. We present the implementation of a related Galerkin method to price call options adjusted with a killing rate in Section~\ref{sec-num}. 
Furthermore, Galerkin-based model reduction techniques have a great potential in financal applications, see Cont, Lantos and Pironneau (2011)\nocite{ContLantosPironneau2011}, \cite{Pironneau2011}, and \cite{SachsSchu2013}, Haasdonk, Salomon and Wohlmuth (2012)\nocite{HaasdonkSalomonWohlmuth2012} and \cite{HaasdonkSalomonWohlmuth2012b}. 


%
%
%
%
%
%
%
%
%

Feynman-Kac representations for viscosity solutions 
 with application to option pricing 
in L\'evy models have been derived in \cite{ContVoltch.2005b} and \cite{ContVoltchkova2005a}. 
Results linking jump processes with Brownian part to variational solutions had already been proven earlier in \cite{BensoussanLions}. 
However, in order to cover some of the most relevant financial models, we have to consider pure jump processes, i.e. processes without a Brownian component, as well.
Pure jump L\'evy models have been shown to fit market data with high accuracy and have enjoyed considerable popularity, see for instance \cite{Eberlein2001}, \cite{Schoutens2003}, \cite{ContTankov.book2004}. Moreover, statistical analysis of high-frequency data supports the choice of pure jump models, see \cite{Ait-SahaliaJacod2014}. 

We realize that pure jump processes differ significantly from processes with a Brownian part. 
The Brownian component translates to a second order derivative in the Kolmogorov operator, while the pure jump part corresponds to an integro differential operator of a lower order of differentiation. 
Accordingly, the second order derivative is only present in Kolmogorov operators of processes with a Brownian component. 
As a consequence, the solution to the Kolmogorov equation of a pure jump L\'evy process does not lie in the Sobolev space $H^1$, the space of quadratic integrable functions with a square integrable weak derivative. Therefore we need a more general solution space. 
In order to make an appropriate choice, recall that L\'evy processes are nicely characterized through the L\'evy-Khinchine formula by the Fourier transform of their distribution or, equivalently, by the symbol. Moreover, the symbol is typically available in terms of an explicit parametric function and as such is the key quantity to parametric L\'evy models. For a wide range of processes, the asymptotic behaviour of the symbol ensures that the solution of the Kolmogorov equation belongs to a Sobolev-Slobodeckii space, i.e.\ it has a derivative of fractional order. 
 Even more, parabolicity with respect to Sobolev-Slobodeckii spaces of the Kolmogorov equations related to L\'evy processes has been characterized in terms of growth conditions on the symbol
in \cite{Glau2013}.
 
So as to allow for typical initial conditions, such as the payout function of a call option in logarithmic variables and the Heaviside step function that relates to distribution functions, we base our analysis more generally on exponentially weighted Sobolev-Slobodeckii spaces. 
We therefore generalize the characterization of parabolicity to time-inhomogeneous L\'evy processes and to exponentially weighted Sobolev-Slobodeckii spaces. In \cite{EberleinGlau2013} existence and uniqueness of weak solutions in exponentially weighted Sobolev-Slobodeckii spaces of Kolmogorov equations related to time-inhomogeneous L\'evy processes and a Feynman-Kac formula has been established. Here, we generalize these results to 
solutions of Kolmogorov equations related to time-inhomogeneous L\'evy processes with possibly discontinuous killing rates. Technically, the present setting is more difficult since the Fourier transform of the solution is not explicitly available and, moreover, the solution is not sufficiently regular for an application of It\^o's formula.  

 The fruitful relation between pseudo differential operators (PDOs) and Markov processes via their symbols has already been extensively used to establish existence of stochastic processes, see for instance the monographs of Jacob from (2001), (2002) and (2005)\nocite{Jacob.I}\nocite{Jacob.II}\nocite{Jacob.III}. For a short overview on the different approaches to construct Feller processes and the use of pseudo differential calculus in this context see Chapter III in the monograph of B\"ottcher, Schilling and Wang (2013)\nocite{BoettcherSchillingWang2013}. 
Let us observe that our question is of a different nature: We establish a Feynman-Kac-type representation of the form \eqref{gl-stochdarst-origin}, while existence of the stochastic processes involved, $L$ and the conditional expectation, are known. An interesting feature of our approach is that we do need not impose growth conditions on the (higher-order) derivatives of the symbol as in the standard symbolic calculus. 
Our approach is more closely related to \cite{Hoh1994}, where a class of martingale problems is solved tracing back the existence of the processes to parabolicity of the Kolmogorov equations with respect to anisitropic Sobolev-Slobodeckii spaces. Compared to the setting in \cite{Hoh1994}, we restrict ourselves to isotropic spaces and constant coefficients, but, more generally, allow for exponentially weighted spaces and possibly discontinuous killing rates.

To comprise all of the requirements, we state our research question more precisely as follows:
 \emph{Under which conditions on the time-inhomogeneous L\'evy process $L$, the possibly discontinuous killing rate $\kappa$, the source $f$ and initial condition $g$ is there \emph{a unique weak solution} in an exponentially weighted Sobolev-Slobodeckii space of Kolmogorov equation~\eqref{parabolic-eq-origin} that allows for a stochastic representation of form~\eqref{gl-stochdarst-origin}?}

To answer our research question, we introduce in the next section the necessary notation and concepts. 
We use this framework first to characterize parabolicity of the Kolmogorov equation in terms of properties of the symbol in Theorem~\ref{Theo-parabolicity}. Prepared thus, we formulate our main result, the  Feynman-Kac-type representation of the weak solution of Kolmogorov equation \eqref{parabolic-eq-origin} in Theorem~\ref{fkac}. In Section~\ref{sec-ex} we find that it is a wide and interesting class of stochastic processes that fall within the scope of this result. Analysing its applications in Section \ref{sec-applications} leads us from typical financial problems further to the characterization of purely probabilistic objects and finally back to the original quantum mechanical ideas of Feynman and Kac---yet in a relativistic guise. Exploiting the advantages of Theorem \ref{fkac} further, we return to its practical realization and implement a Galerkin scheme to solve Kolmogorov equation\tild \eqref{parabolic-eq-origin} in Section \ref{sec-num}. We find that thanks to Theorem~\ref{fkac} the solutions obtained thus correspond to option prices. 
With the numerical implementation at hand, we visualize and discuss the effect of killing rates of indicator type.  
 Section~\ref{app-robust} presents a robustness result for weak solutions that is required in our proof of Theorem~\ref{fkac} in Section \ref{sec-proof-fkacgeneral}. In this last section  we also identify desirable regularity properties of the solutions to the Kolmogorov equation. Appendix \ref{sec-adjop} provides two technical lemmata for the symbol and the operator, and Appendix~\ref{sec-proof-Theo-parabolicity} concludes with the proof of Theorem~\ref{Theo-parabolicity}.

 \section{Preliminaries and notation}
 
 In order to present the main result of the present article, we first introduce the underlying stochastic processes, the Kolmogorov equation with killing rate, its weak formulation as well as the solution spaces of our choice. We denote by $C_0^\infty(\rrd)$ the set of smooth real-valued functions with compact support in $\rrd$ and let
\begin{equation}\label{def_FT}
\OF(\varphi):= \int_{\rrd}\ee{i\skl \xi,x\skr} \varphi(x) \dd x
\end{equation}
be the Fourier transform of $\varphi\in C^\infty_0(\rrd)$ and $\OF^{-1}$ be its inverse.

Let a stochastic basis $(\Omega,\OF, (\OF_t)_{0\le t\le T}, P)$ be given and let
$L$ be an $\rrd$-valued \emph{time-inhomogeneous L\'evy process} with characteristics $(b_t,\sigma_t,F_t;h)_{t\ge0}$. That is $L$ has independent increments and for fixed $t\ge0$ its characteristic function is given by
\begin{align}\label{eq-charPIIAC}
E\ee{i \skl \xi, L_t \skr} = \ee{ -\int_0^t A_s(-i\xi)\dd s}\quad \text{for every }\xi\in\rrd,
\end{align}
where, for every $t\ge 0$ and $\xi\in\rrd$, the \emph{symbol of the process }is defined as
\begin{equation}\label{A_t}
A_t(\xi) := \frac{1}{2}\langle \xi,\sigma_t \xi\rangle + i\langle \xi,b_t\rangle 
- \int_{\rrd}\left(\ee{-i\langle \xi,y\rangle} -1+ i\langle \xi,h(y)\rangle\right)\,F_t(\dd y).
\end{equation}
Here, for every $s>0$, $\sigma_s$ is a symmetric, positive semi-definite $d\times d$-matrix, $b_s\in \rr^d$, and $F_s$ is a L\'evy measure, i.e. a positive Borel measure on $\rrd$ with $F_s(\{0\})=0$ and $\int_{\rr^d} (|x|^2 \wedge 1) F_s(\dd x)  < \infty$. Moreover, $h$ is a truncation function i.e. $h:\rr^d\to\rr$ such that $\int_{\{|x|>1\}} h(x) F_t(\dd x)<\infty$ with $h(x)=x$ in a neighbourhood of $0$. We assume the maps $s\mapsto \sigma_s$, $s \mapsto b_s$ and $s\mapsto\int (|x|^2\wedge 1) F_s(\dd x)$ to be Borel-measurable with, for every $T>0$,
\begin{equation}\label{int-PIIAC}
\int_0^{T} \Big( |b_s| + \|\sigma_s\|_{\OM(d\times d)} +\int_{\rr^d} (|x|^2 \wedge 1) F_s(\dd x) \Big) \dd s < \infty,
\end{equation}
 where $ \|\cdot\|_{\OM(d\times d)}$ is a norm on the vector space formed by the $d\times d$-matrices. 
\\
The  \emph{Kolmogorov operator of the process} $L$ is given by
\begin{align}\label{def-A}
\begin{split}
\OA_t \varphi(x)\coloneqq & - \frac{1}{2}\sum_{j,k=1}^d \sigma^{j,k}_t\frac{\partial^2 \varphi}{\partial x_j\partial x_k}(x)-\sum_{j=1}^d b^j_t\frac{\partial \varphi}{\partial x_j}(x)\\
&-\int_{\rr^d}\Big( \varphi(x+y)-\varphi(x)-\sum_{j=1}^d\frac{\partial \varphi}{\partial x_j}(x) \,  h_j(y)\Big)F_t(\dd y)
\end{split}
\end{align}
for every $\varphi\in C^\infty_0(\rrd)$, where $h_j$ denotes the $j$-th component of the truncation function $h$. By some elementary manipulations we obtain
\begin{equation}\label{eq-Aispseudo}
\OA_t \varphi=\OF^{-1}(A_t \OF(\varphi))\qquad\text{for all }\varphi\in C^\infty_0(\rrd),
\end{equation}
which shows us that the Kolmogorov operator $\OA$ is a pseudo differential operator with symbol~$A$.

Following the classical way to define solution spaces of parabolic evolution equations, we introduce a \emph{Gelfand triplet} $(V,H,V^\ast)$, which consists of a pair of separable Hilbert spaces $V$ and $H$ and the dual space $V^\ast$ of $V$ such that there exists a continuous embedding from $V$ into $H$.
We then denote by $L^2\big(0,T; H\big)$ the space of weakly measurable functions $u:[0,T]\to H$ with $\int_0^T\|u(t)\|_H^2 \dd t < \infty$ and by $\partial_t u$ the derivative of $u$  with respect to time in the distributional sense.
The Sobolev space 
\begin{equation}\label{def-W1}
W^1( 0,T; V,H) := \Big\{ u\in L^2\big(0,T;V\big) \,\Big| \,\partial_t u\in L^2\big(0,T; V^\ast\big) \Big\},
\end{equation}
will serve as solution space for equation \eqref{parabolic-eq-origin}. For a more detailed introduction to the space $W^1\big(0,T; V, H\big)$, which relies on the Bochner integral, we refer to Section 24.2 in \cite{Wloka-english}.  More information on Gelfand triplets can be found for instance in Section 17.1 in \cite{Wloka-english}.

Usually, variational equations of a similar type as the heat equation are formulated with respect to Sobolev spaces, and thus are based on both $H^1$ and $L^2$. Since we include pure jump processes in our analysis, operator \eqref{def-A} may be of fractional order. We therefore work with Sobolev-Slobodeckii spaces, which formalize the notion of a derivative of fractional order. Turning to a typical financial problem, we express the price of a call option as solution to a Kolmogorov equation of type~\eqref{parabolic-eq-origin}. We then obtain $\kappa=0$ and $f=0$, while the initial condition is given by $g(x)=(S_0\ee x - K)^+$. We now have to realize that the initial condition $g\notin L^2(\rrd)$ and we cannot use an $L^2$-based approach. The exponentially dampened function, $x\mapsto g(x) \ee{\eta x}$, though belongs to $L^2(\rrd)$ for every $\eta<-1$.
 Thus, in order to incorporate initial conditions that typically arise in financial problems, we allow for an exponential weight.  We further increase the class of function spaces by a domain splitting argument, see Remark~\ref{rem-split-orthants} on page \pageref{rem-split-orthants}. 

To make these considerations formally precise, we define the \emph{exponentially weighted Sobolev-Slobodeckii space} $H^\alpha_\eta(\rrd)$ with index $\alpha\ge0$ and weight $\eta\in \rrd$ as the completion of $C_0^\infty(\rrd)$ with respect to the norm $\|\cdot\|_{H^\alpha_\eta}$ given by
\begin{equation}\label{def_normHseta}
\|\varphi\|_{H^\alpha_\eta}^2:= \int_{\rrd} \big(1+|\xi|\big)^{2\alpha}\big|\OF(\varphi)(\xi - i \eta)\big|^2 \dd \xi.
\end{equation}
Observe that this is a separable Hilbert space.
For $\eta=0$ the space $H^\alpha_\eta(\rrd)$ coincides with the Sobolev-Slobodeckii space $H^\alpha(\rrd)$ as it is defined e.g. in \cite{Wloka-english}. For $\alpha=0$ the space $H^\alpha_\eta(\rrd)$ coincides with the weighted space of square integrable functions
$L^2_\eta(\rr^d) \coloneqq \big\{ u\in L^1_{\loc}(\rr^d)\, \big| \, x\mapsto u(x)\ee{\langle \eta, x\rangle } \in L^2(\rr^d) \big\}$. Furthermore, we denote the dual space of $H^\alpha_\eta(\rrd)$ by $\big(H^\alpha_\eta(\rrd))^\ast$.

%


Let $a: [0,T]\times H^\alpha_\eta(\rrd) \times H^\alpha_\eta(\rrd) \to \rr$ be a family $(a_t)_{t\in[0,T]}$ of bilinear forms that are measurable in $t$ with associated linear operators $\OA_t: H^\alpha_\eta(\rrd) \to \big(H^\alpha_\eta(\rrd))^\ast$ given by
\begin{equation}\label{rel_OAa}
\OA_t(u)(v) = a_t(u,v)\qquad\text{for all } u,v\in  H^\alpha_\eta(\rrd)
\end{equation}
and whose related symbols $A_t:\rrd\to \cc$ are such that
\begin{equation}\label{rel_OAA}
\OA_t(\varphi) = \frac{1}{(2\pi)^d} \int_{\rrd}\ee{-i\skl \xi, x\skr} A_t(\xi)\OF(\varphi)(\xi) \dd \xi \qquad \text{for all }\varphi \in C^\infty_0(\rrd).
\end{equation}
We close the section with the weak formulation of Kolmogorov equation~\eqref{parabolic-eq-origin}. 

\begin{defin}
Let $V=H^\alpha_\eta(\rrd)$ and $H=L^2_\eta(\rrd)$, $\kappa:[0,T]\times \rrd \to\rr$ measurable and bounded,  $f\in L^2\big(0,T; V^\ast\big)$ and $g\in H$. Then 
$u\in W^1( 0,T; V,H)$ is a \emph{weak solution} of Kolmogorov equation \eqref{parabolic-eq-origin}, if for almost every $t\in(0,T)$, 
\begin{equation}\label{def-para}
\scalebox{.93}[1]{$\displaystyle  \skl \partial_t u(t), v\skr_{H} + a_{T-t}( u(t), v) + \skl \kappa_{T-t} u(t), v\skr_{H}  =\, \skl f(t) | v\skr_{V^\ast\times V}\quad \text{for all }v\in V$}
\end{equation}
and $u(t)$ converges to $g$ for $t\downarrow0$ in the norm of $H$.
\end{defin}

\section{Main results}\label{sec-main}
Equipped with the necessary notation and concepts, we now focus on our main purpose, providing a Feynman-Kac formula linking weak solutions of PIDEs with killing rates to conditional expectations.

Following a classical way to prove existence and uniqueness of weak solutions of a parabolic equation, we verify continuity and a G{\aa}rding inequality of its bilinear form. We specify the notion of parabolicity accordingly and adapt it to our framework:

\begin{defin}
Let $\OA$ be an operator associated with bilinear form $a$.

We say $\OA$, respectively $a$, is \emph{parabolic} with respect to $H^{\alpha/2}_\eta(\rrd), L^2_\eta(\rrd)$, if  
there exist constants $C,G>0$, $G'\ge0$ such that uniformly for all $t\in[0,T]$ and all $u,v\in H^{\alpha/2}_\eta(\rrd)$,
\begin{align}
\big|a_t(u,v)\big| &\le C \|u\|_{H^{\alpha/2}_\eta(\rrd)}\|v\|_{H^{\alpha/2}_\eta(\rrd)} \tag{Continuity (Cont-$a$)}\label{cont-a}\\
a_t(u,u)&\ge G \|u\|_{H^{\alpha/2}_\eta(\rrd)}^2 - G'\|u\|_{L^2_\eta(\rrd)}^2.\tag{G{\aa}rding inequality (G{\aa}rd-$a$)}\label{gard-a}
\end{align}

For $R\subset\rrd$,
we say that the parabolicity of $\OA$, respectively $a$, (with respect to $\big(H^{\alpha/2}_\eta(\rrd), L^2_\eta(\rrd)\big)_{\eta\in R}$) is \emph{uniform in $[0,T]\times R$}, if for all $u,v\in \cup_{\eta\in R}H^{\alpha/2}_\eta(\rrd)$ the mapping $t\mapsto  a_t(u,v)$ is c\`adl\`ag and 
there exist constants $C,G>0$, $G'\ge0$ such that uniformly for all $\eta \in R$, all $t\in[0,T]$ and $u,v\in H^{\alpha/2}_\eta(\rrd)$ inequalities (Cont-$a$) and (G{\aa}rd-$a$) are satisfied.

\end{defin}

As highlighted in equation \eqref{eq-Aispseudo}, the Kolmogorov operator of a time-inhomo\-geneous L\'evy process is a pseudo differential operator. Its symbol is explicitly known for various classes and in general is characterized by the exponent of the L\'evy-Khinchine representation. 
Therefore, we express our main assumptions in terms of the symbol of the process. For L\'evy processes with symbols $A$, it has been shown in \cite{Glau2013}, Theorem 3.1, that 
the corresponding bilinear form is parabolic with respect to $H^{\alpha/2}(\rrd)$, $L^2(\rrd)$ if and only if constants $C,G,G'>0$ and $0\le \beta<\alpha$ exist such that for every $\xi\in\rrd$,
\begin{align}
\big|A(\xi) \big|&\le  C \big( 1 + |\xi| \big)^{\alpha}\label{cont-Aalt}\\
\Re\big(A(\xi) \big)&\ge G \big(1 + |\xi| \big) ^\alpha -  G' \big(1 + |\xi| \big) ^\beta.\label{Gard-Aalt}
\end{align}
We generalize this growth condition so as to render it suitable for the setting of time-inhomo\-geneous L\'evy processes and weighted Sobolev-Slobodeckii spaces.
%
We find that an extension of the bilinear form to weighted Sobolev-Slobodeckii spaces corresponds to a shift of the symbol in the complex plane. Symbols can be extended to complex domains if the appropriate exponential moment condition is satisfied. 
Let $L$ be a time-inhomogeneous L\'evy process.
First notice that $L_t$ is infinitely divisible with L\'evy measure $\widetilde{F_t}(\dd x):= \int_0^tF_s(\dd x)\dd s$ for every $t\in[0,T]$, as has been shown by \cite{EberleinKluge06a}, Lemma\tild 1. 
Theorem 25.17 in \cite{Sato} now
implies that, for all $\eta\in\rrd$,
 \begin{align}\label{EMeta}
\int_0^T\int_{|x|>1}\ee{\skl \eta,x\skr}F_t(\dd x)\dd t<\infty \tag{$EM(\eta)$}
\end{align}
  is equivalent to the exponential moment condition $E\big[\ee{\skl \eta , L_T\skr}\big] <\infty$ and 
\begin{align}\label{Phuteta}
E\big[\ee{\skl i\xi +\eta , L_t\skr}\big] = \ee{-\int_0^t A_s(-\xi+i\eta)\dd s}\quad\text{for all }\xi\in\rrd \text{ and }t\ge0.
\end{align}
We therefore formulate the conditions in terms of an exponential moment condition on the process and growth conditions on the symbol extended to a complex domain. It turns out that this complex domain can conveniently be chosen as a tensorized complex strip.
More precisely, for weight $\eta=(\eta_1,\ldots,\eta_d)$,~let
\begin{align}\label{Ueta}
U_{\eta} &\coloneqq  \big\{ z\in \ccd \, \big|\, \Im(z_j) \in \{0\}\cup\sign(\eta_j) [0,|\eta_j|) \,\text{for }j=1,\ldots,d \big\},\\
R_\eta &\coloneqq \sign(\eta_1)[0,|\eta_1|]\times\cdots\times\sign(\eta_d)[0,|\eta_d|].
\end{align}
From Theorem 25.17 in \cite{Sato}, we also know that the complex set on which $A_s$ is definable is convex. 
Lemma 2.1 (c) in \cite{EberleinGlau2013} shows for the present setting  that the map $z\mapsto A_t(z)$ has a continuous extension to the complex domain $\overline{U_{-\eta}}$ that is analytic in the interior $\Ukri_{-\eta}$.

We will derive the main results related to the following set of conditions.
\begin{conditions}\label{conds}
\emph{For weight $\eta\in\rrd$ and index $\alpha\in(0,2]$, let $A=(A_t)_{t\in[0,T]}$ be a symbol with extension to $\overline{U_{-\eta}}$ and, if available, let $L$ denote the time-inhomogeneous L\'evy process with symbol $A$.
\vspace{-0.5ex}
\begin{enumerate}[leftmargin=3em, label=(A\arabic{*}),widest=(A4)]\label{A2-A4}
\item
For every $\eta' \in R_{\eta}$,
\begin{align}\label{EMR-eta}
E[\ee{-\skl\eta',L_t\skr}]&<\infty.
\tag{Exponential moment condition $(EM)$}
\end{align}
\item
There exists a constant $C>0$ such that uniformly for all $\eta'\in R_{\eta}$ and $t\in[0,T]$,
\begin{align*}
 \big|A_t(\xi - i\eta') \big|&\le  C \big( 1 + |\xi| \big)^{\alpha}. \tag{Continuity condition (Cont-$A$)}\label{cont-
A}
\end{align*}
\item
There exist constants $G>0$, $G' \ge 0$ and $0\le \beta<\alpha$ such that uniformly for all $\eta'\in R_{\eta}$ and $t\in[0,T]$,
\begin{align*}
\scalebox{.84}[1]{$\displaystyle \Re\big(A_t(\xi -i\eta') \big)$} \ge \scalebox{.84}[1]{$\displaystyle G \big(1 + |\xi| \big) ^\alpha -  G' \big(1 + |\xi| \big) ^\beta.$}
\tag{G{\aa}rding condition (G{\aa}rd-$A$)}\label{gard-A}
\end{align*}
\item
For every fixed $\eta'\in R_{\eta}$ and $\xi\in\rrd$ the mapping $t\mapsto A_t(\xi-i\eta')$ 
is c\`adl\`ag. 
\end{enumerate}
}
\end{conditions}
We say that $A$ has \emph{Sobolev index} $\alpha$  \emph{uniformly in} $[0,T]\times R_{\eta}$, if $A$ has an extension to $\overline{U_{-\eta}}$ that satisfies (A2)--(A4). If $A$ is the symbol of process $L$, we also say $L$ has \emph{Sobolev index} $\alpha$  \emph{uniformly in} $[0,T]\times R_{\eta}$.\\

Conditions (A1)--(A4) are satisfied for a large set of processes, for instance for tempered stable and normal inverse Gaussian processes as well as their time-inhomogeneous extensions. In Section \ref{sec-ex} we look in detail at the verification of the conditions.


Notice that for $\eta=0$ we have $R_\eta=\{0\}$. Thus (A1) is trivially satisfied and (A2)--(A3) simplify accordingly. This case corresponds to the case of Sobolev-Slobodeckii spaces without weighting and is covered by the following results.
If, moreover, the symbol is constant in time, (A4) is trivially satisfied and (A1)--(A4) reduce to \eqref{cont-Aalt} and \eqref{Gard-Aalt}. 
%
Conditions (A1)--(A3) were introduced in \cite{EberleinGlau2013} to show existence and uniqueness of weak solutions of the related Kolmogorov equation (without killing rate) along with a Feynman-Kac-type formula with application to European option prices in time-inhomogeneous L\'evy models. 
We additionally require (A4), which only imposes a mild technical restriction.

Our framework defined, let us now state our main results.
We first show in Theorem \ref{Theo-parabolicity} the  equivalence between parabolicity with respect to weighted Sobolev-Slobodeckii spaces and growth conditions (A2) and (A3), thereby generalizing the result for Sobolev-Slobodeckii spaces and conditions \eqref{cont-Aalt} and \eqref{Gard-Aalt} to the present setting.
The characterization of (uniform) parabolicity in terms of conditions on the symbol is interesting in its own right.  
It is, moreover, one of the key steps in our proof of Theorem \ref{fkac} below, which establishes a Feynman-Kac-type representation.

\begin{theorem}\label{Theo-parabolicity}
For $\eta\in\rrd$ and $\alpha\in(0,2]$, let $L$ be a time-inhomogeneous L\'evy process satisfying exponential moment condition (A1). Then the following two assertions are equivalent.
\begin{enumerate}[label=$(\roman*)$,leftmargin=2em]
\item
The Kolmogorov operator of $L$ is uniformly parabolic in $[0,T]\times R_{\eta}$ with respect to 
$\big(H^\alpha_{\eta'}(\rrd), L^2_{\eta'}(\rrd)\big)_{\eta'\in R_{\eta}}$.
\item
$L$ has Sobolev index $2\alpha$ uniformly in $[0,T]\times R_{\eta}$.
\end{enumerate}
\end{theorem}

The proof of Theorem \ref{Theo-parabolicity} is given in Appendix~\ref{sec-proof-Theo-parabolicity}, where, moreover, Theorem~\ref{Prop-parabolicity} provides a more general version of this result for operators and symbols that are not necessarily related to stochastic processes.

Assume (A1)--(A4). 
Then Theorem~\ref{Theo-parabolicity} shows in particular the parabolicity of the related bilinear form uniformly in time with respect to $H^{\alpha}_\eta(\rrd)$ and $L^{2}_\eta(\rrd)$. 
Now the classical existence and uniqueness result, see for instance 
 Theorem~23.A in \cite{Zeidler}, gives us that Kolmogorov equation~\eqref{parabolic-eq-origin} has a unique weak solution $u$ in the space $W^1\big(0,T;H^{\alpha}_\eta(\rrd),L^2_\eta(\rrd)\big)$.


We now turn to the stochastic representation of this solution.
For an integrable or nonnegative random variable $X$ we denote
 \begin{equation}
 E_{0,x}(X) := E_x(X), \quad E_{t,x}(X) := E(X|L_t=x)\,\,\text{for }t>0,\quad 
 \end{equation}
where $x\mapsto E(X|L_t=x)$ is the factorization of the conditional expectation $E(X|L_t)$ and $E_x$ the expectation with respect to the probability measure $P_x$ such that $P_x(L_0=x)=1$.
\begin{theorem}\label{fkac}
For $\eta\in\rrd$ and $\alpha\in(0,2]$, let $L$ be an $\rrd$-valued time-inho\-mo\-geneous L\'evy process with symbol $A=(A_t)_{t\in[0,T]}$ that satisfies (A1)--(A4). 
Then 
\begin{enumerate}[label=$(\roman*)$,leftmargin=2em]
\item\label{item:fkaci} for $\kappa:[0,T]\times \rrd \to\rr$ measurable and bounded, $f\in L^2\big(0,T ; (H^{\alpha/2}_\eta(\rrd))^\ast \big)$ and $g\in L^2_\eta(\rrd)$ Kolmogorov equation \eqref{parabolic-eq-origin} has a unique weak solution $u\in W^1\big(0,T;H^{\alpha/2}_\eta(\rrd),L^2_\eta(\rrd)\big)$;
\item\label{item:fkacii} if, additionally, $f\in L^2\big(0,T ; H^l_\eta(\rrd)\big)$ for some $l\ge 0$ with $l>(d-\alpha)/2$, then, for every $t\in[0,T]$ and a.e. $x\in\rrd$,
\begin{equation}\label{gl-stochdarst}
 \begin{split}
  u(T-t,x)
&=
E_{t,x}\Big(g(L_T)\ee{-\int_t^T   \kappa_h(L_{h}) \dd h} \\
&
\quad\, + \int_t^{T} \! f(T-s,L_s)\ee{-\int_t^s   \kappa_h(L_{h}) \dd h} \dd s\Big).
 \end{split}
\end{equation}
\end{enumerate}
\end{theorem}
As we have already seen, part \ref{item:fkaci} of Theorem \ref{fkac} follows from Theorem \ref{Theo-parabolicity} and the classical existence and uniqueness result for solutions of parabolic equations.
Part \ref{item:fkacii} is considerably more involved and Section \ref{sec-proof-fkacgeneral} is devoted to its proof.

The major benefit of Theorem \ref{fkac} for financial applications is that conditional expectations of form~\eqref{gl-stochdarst}, which naturally appear as derivatives and asset prices, are now characterized by weak solutions of PIDEs. Therefore, the prices can be computed by numerically solving an equation of form~\eqref{parabolic-eq-origin}. So as to illustrate the method and the effect of killing rates, we present among others an application to employee options in Section~\ref{sec-applications} and provide its Galerkin discretization in Section~\ref{sec-num}.

Theorem \ref{fkac} furthermore shows a specific type of regularity of conditional expectation \eqref{gl-stochdarst}.   It is interesting to identify sufficient conditions for H\"older continuity. Theorem 8.2 in Nezza, Palatucci
and Valdinoci (2011)\nocite{HitchhikersGuide} provides the appropriate Sobolev embedding result. 
Thus we obtain as an immediate consequence of Theorem \ref{fkac} the following corollary.

\begin{corollary}\label{cor-hoelder}
Under the assumptions and notations of Theorem \ref{fkac} in the univariate case, i.e. for $d=1$, for $\alpha\in(1,2]$ and any fixed $t\in(0,T)$, the function $x\mapsto u(t,x)$ is $\lambda$-H\"older continuous with $\lambda=\frac{\alpha-1}{2}$
, i.e. 
\[
\sup\limits_{x,y\in\rr, x\neq y}\frac{|u(t,x)-u(t,y)|}{|x-y|^\lambda}<\infty.
\]
In particular, $x\mapsto u(t,x)$ is continuous and equality \eqref{gl-stochdarst} in Theorem \ref{fkac} holds for every $x\in\rr$.
\end{corollary}
%
%


\section{Examples of classes of time-inhomogeneous L\'evy processes}\label{sec-ex}
Let us explore the nature of Conditions (A1)--(A4) and show that they are satisfied for a wide class of processes. Conditions (A1)--(A4) naturally apply to processes that are specified through their symbol. Notice that the symbol is expressed in terms of the characteristics of the process. 
We exploit this in Proposition \ref{cond-onF} to establish concrete accessible conditions for real valued time-inhomogeneous pure jump L\'evy processes with absolutely continuous L\'evy measures, while
Proposition \ref{prop-jumpdiff} treats time-inhomogeneous multivariate jump diffusions.

We should, however, realize that Conditions (A1)--(A4) 
are not satisfied by all L\'evy processes. On the one hand, continuity and G{\aa}rding condition (A2) and (A3) have implications for the distributional properties of the process:
 
\begin{remark}\label{rem-implicationsA1A3}
Fix  $\eta\in\rrd$ and $\alpha\in(0,2 ]$, and
let $L$ be a time-inhomogeneous L\'evy process  with symbol $A=(A_t)_{t\ge0}$. If G{\aa}rding condition (A3) is satisfied for weight $\eta$ and index $\alpha$, then there exist $C_1,C_2>0$ such that uniformly for all $\eta'\in R_{\eta}$ and $0\le s\le t \le T$,
\begin{align}\label{absch-Phuteta}
\big|\ee{-\int_{s}^t A_u(\xi-i\eta') \dd u}\big|\le C_1 \ee{-(t-s)C_2|\xi|^{\alpha}}.
\end{align}
In particular, (A3) implies for every $t\in(0,T]$ that the distribution of $L_t$ has a smooth Lebesgue density.
\end{remark}

On the other hand, 
continuity and G{\aa}rding condition  (A2) and (A3) 
 relate to the path behaviour of the process:
If a L\'evy process with symbol $A$ satisfies (A2) and (A3) for $\alpha\in(0,2)$ and $\eta=0$, then $\alpha$ is its Blumenthal-Getoor index, as shown in \cite{Glau2013}, Theorem 4.1. Hence, every pure jump L\'evy process satisfying assumptions (A2) and (A3) has infinite jump activity. On this basis we may for instance conclude that compound Poisson processes do not satisfy~(A3).

Variance Gamma processes have Blumenthal-Getoor index $0$ and thus do not satisfy both (A2) and (A3), as noticed in part (iv) of Example~4.1 in \cite{Glau2013}.
However, pure jump L\'evy processes can be approximated by a sequence of L\'evy jump diffusion processes with nonzero Brownian part. This can always be achieved by adding a diffusion part and letting its volatility coefficient tend to zero. Example \ref{ex-jumpdiff} shows that pure jump L\'evy processes can be approximated by L\'evy processes for which (A1)--(A3) are satisfied for weight $\eta=0$ and index $\alpha=2$. \cite{AsmussenRosinski2001} provide a sequence with better approximation properties for Monte Carlo techniques, which could be exploited further. 

Before continuing the discussion on the validity of Conditions (A1)--(A4) 
 for other classes of processes, we observe the following. 
\begin{remark}
For $\eta\in\rrd$ and $\alpha\in(0,2]$, let $A$ be the  symbol of a L\'evy process that satisfies exponential moment condition (A1).  By virtue of Lemma \ref{lem-Aeta} in Appendix \ref{sec-adjop} and the continuity of symbols of L\'evy processes (as mappings from $\rrd$ to $\cc$), the validity of continuity condition (A2) for $A$ is equivalent to the following \emph{asymptotic condition}: For every $N>0$ there exist a constant $G>0$ such that for every $\eta'\in R_{\eta}$,
\begin{align}
\Re\big(A(\xi - i\eta') \big) &\ge G|\xi|^\alpha -  A(i\eta')\quad\text{for every }\xi\in\rrd\,\,\text{such that }|\xi|>N.
\end{align}
\end{remark}
We devote the remainder of this section to providing sufficient conditions for the validity of Conditions (A1)--(A4) for time-inhomogeneous jump diffusions, for pure jump L\'evy processes, and for time-inhomogeneous processes.


\subsection{Jump diffusions}

For time-inhomogeneous L\'evy jump diffusion processes we find that Conditions (A1)--(A4) are satisfied under remarkably weak conditions:

\begin{proposition}\label{prop-jumpdiff}
Fix some $\eta\in\rrd$.
 Let $L$ be a time-inhomo\-geneous L\'evy process with characteristics $(b_t,\sigma_t,F_t;h)_{0\le t\le T}$ such that
 \begin{align}
\sup_{t\in[0,T]}\int_{|x|>1}\ee{-\skl\eta',x\skr} F_t(\dd x)<\infty \quad \text{ for every $\eta'\in R_{\eta}$ and}\label{uniform-A1}\\
\sup_{t\in[0,T]}\bigg\{|b_t| + \|\sigma_t^{-1}\| +  \|\sigma_t\| + \int_{\rrd} \big(|x|^2\wedge 1\big) F_t(\dd x)\bigg\}<\infty.\label{supbsigmaF}
\end{align}
Then (A1)--(A3) are satisfied for weight $\eta$ and index $\alpha=2$.
\end{proposition}
\begin{proof}
Due to the equivalence of $EM(\eta)$ and the exponential moment condition, \eqref{uniform-A1} implies (A1). Observing that
\begin{align*}
\Re\big(A_t(\xi-i\eta')\big) =& \skl b_t, \eta'\skr + \frac{1}{2}\skl \eta',\sigma_t\eta'\skr + \int_{\rrd} \!\!\Big( (\skl h(x),\eta'\skr -1)\ee{-\skl\eta',x\skr} - 1\Big)F_t(\dd x)\\
& + \frac{1}{2}\skl \xi,\sigma_t\xi\skr + \int_{\rrd}\Big(\cos\big(\skl\xi,x\skr \big) - 1 \Big)F_t(\dd x)
\end{align*}
and $\int_{\rrd}\big(\cos(\skl\xi,x\skr) - 1 \big)F_t(\dd x)\ge0$,
inequalities \eqref{uniform-A1} and \eqref{supbsigmaF} yield G{\aa}rding condition (A3) with $\alpha=2$. Similarly, inequalities \eqref{uniform-A1}, \eqref{supbsigmaF} and
\begin{align*}
\Im\big(A_t(\xi-i\eta')\big) =& \skl b^{-\eta'}_t, \xi \skr + \int_{\rrd}\Big(\sin\big(\skl\xi,x\skr \big) - \skl \xi, h(x) \skr \Big)\ee{ - \skl\eta',x\skr}F_t(\dd x),
\end{align*}
where $b^{-\eta'}_t=b + \sigma \cdot \eta' + \int_{\rrd} \big( \ee{\skl \eta',y\skr} - 1 \big) h(y) F(\dd y)$ as defined in Lemma \ref{lem-Aeta} in Appendix \ref{sec-proof-Theo-parabolicity}, yield continuity condition (A2), which concludes the proof.\qed
\end{proof}

For L\'evy jump diffusion processes the conditions simplify considerably:

\begin{example}[Multivariate L\'evy processes with Brownian part]\label{ex-jumpdiff}
Fix $\eta\in\rrd$ and let $L$ be an
 $\rrd$-valued L\'evy pro\-cesses with characteristics $(b,\sigma,F;h)$ such that $\sigma$ is a positive definite matrix  and the L\'evy measure $F$ satisfies $\int_{|x|>1} \ee{- \eta x}F_t(\dd x)<\infty$. Then (A1)--(A3) hold for  weight $\eta\in\rrd$ and index $\alpha=2$.
\end{example}


In order to verify those assumptions of Proposition \ref{prop-jumpdiff} that concern the pure jump part of the process, it suffices to consider the pure jump processes separately, as the following Lemma shows.



\begin{lemma}\label{lem-symbolconstructions}
For $j=1,2$, let $L^j$ be two stochastically independent time-inhomo\-geneous L\'evy processes with symbol $A^j$ such that (A1)--(A4) are satisfied for the same weight $\eta\in\rrd$ and the possibly different indices $\alpha^j$. Then the sum $L:=L^1+L^2$ is a time-inhomogeneous L\'evy process with symbol $A:=A^1+A^2$, and  (A1)--(A4) are satisfied for  weight $\eta$ and index $\alpha:=\max(\alpha^1,\alpha^2)$.
\end{lemma}
Lemma \ref{lem-symbolconstructions} generalizes Remark 4.1. in \cite{Glau2013} to the case where $\eta\neq 0$ and we omit its elementary proof.

\subsection{Pure jump L\'evy processes and operators of fractional order}


We now consider a class of multivariate processes, which frequently occurs in finance and whose symbol is explicitly given. 

%
\begin{example}[Multivariate Normal Inverse Gaussian (NIG) processes]\label{ex-multidNIG}
Let $L$ 
be an $\rr^d$-valued NIG-process, i.e.\ a L\'evy process such that
$L_1=(L^1_1,\ldots,L^d_1)\sim$ $\text{NIG}_d(\tilde{\alpha},\beta,\delta,\mu,\Delta)$
for parameters
$\tilde{\alpha},\delta\ge0$, $\beta,\mu\in\rr^d$ and symmetric positive definite matrix $\Delta\in\rr^{d\times d}$ with $\tilde{\alpha}^2>\langle \beta, \Delta\beta \rangle$.
The symbol of $L$ is given by
\begin{align*}
A(u)&=  i \langle u,\mu\rangle 
       -  \delta\Big(\sqrt{\tilde{\alpha}^2-\langle \beta,\Delta\beta\rangle}
           -\sqrt{\tilde{\alpha}^2-\langle \beta+iu,\Delta(\beta+iu)\rangle}\Big),
\end{align*}
where we denote by $\skl \cdot,\cdot\skr$ the product $\skl z,z'\skr = \sum_{j=1}^d z_j z_j'$ for $z\in\ccd$. Compare e.g. equation (2.3) in \cite{Hammerstein}.
\\
Assumptions (A1)--(A3) are satisfied for index $\alpha=1$ and every $\eta\in\rrd$ such that $\tilde{\alpha}^2> \skl \beta + \eta', \Delta(\beta + \eta')\skr$ for all $\eta'\in R_{\eta}$. This is in particular the case, if
\begin{align}\label{parameters-nig}
\|\beta\|^2 + \|\eta\|^2 \le \tilde{\alpha}^2/\|\Delta\|\text. 
\end{align}
To summarize, if the parameters of $L$ satisfy \eqref{parameters-nig}, Conditions (A1)--(A4) are satisfied for weight $\eta$ and Sobolev index $1$.

\end{example}

Since pure jump L\'evy processes can be defined through a L\'evy measure and a constant drift, we are interested in finding conditions on both the L\'evy measure and the drift that imply Conditions (A1)--(A4). Let us address this issue for real-valued time-homogeneous pure jump L\'evy processes whose L\'evy measure is absolutely continuous. For this class we generalize Proposition 4.2 in \cite{Glau2013} to time-inhomogenuity and weights $\eta\neq0$. Thereby we obtain explicit conditions on the characteristics that imply Conditions (A1)--(A4).
\begin{conditions}\label{cond-onF}
Fix $\eta\in\rr$ and $\alpha\in(0,2]$, 	and let $L$ be a real-valued time-inhomogeneous L\'evy process with characteristics $(b_t,\sigma_t,F_t;h)_{t^\ge0}$.
\vspace{-0.5ex}
\begin{enumerate}[leftmargin=3em, label={\rm(F\arabic{*})},widest=(F4)]\label{f1-f4}
\item
$\int_0^T\int_{|x|>1} \ee{- \eta x}F_t(\dd x)\dd t<\infty$,
\item
$F_t$ is absolutely continuous for every $t\in[0,T]$ with density $f_t$, i.e.\ $F_t(\dd x)= f_t(x) \dd x$. Denote the symmetric part by $f^{sym}_t(x):=\big(f_t(x) + f(-x)\big)/2$ and the antisymmetric part by $f^{asym}_t(x) := f_t(x) - f^{sym}_t(x)$. 
\item
There exist constants $C_1,C_2,\epsilon>0$ and $0\le \beta<\alpha<2$ and a function $g:[0,T]\times[-\epsilon,\epsilon]\to \rr$ such that uniformly for all $t\in[0,T]$,
\begin{align*}
f^{sym}_t(x) &\le \frac{C_1}{|x|^{1+\alpha}} + g(t,x)
\,\text{and }
\big|g(t,x)\big| \le \frac{C_2}{|x|^{1+\beta}}\quad\text{for all }|x|<\epsilon.
\end{align*}
\item
If $\alpha=1$, there exist constants $C_3,\epsilon>0$ and $\beta\in(0,1]$ 
such that uniformly for all $t\in[0,T]$,
\begin{align}\label{eq-asym}
\big|f^{asym}_t(x)\big| \le  \frac{C_3}{|x|^{1+\beta}}\quad\text{for all }|x|<\epsilon.
\end{align}
 If $\alpha<1$, then inequality \eqref{eq-asym} holds for some $\beta\in[0, \alpha]$  and, moreover, $b_t= \int h(x) F_t(\dd x)$ for every $t\in[0,T]$.
\end{enumerate}
\end{conditions}

\begin{proposition}\label{prop-cond-F}
Let $L$ be a real-valued time-inhomogeneous pure jump L\'evy processes with characteristics $(b_t,0,F_t)_{t^\ge0}$. Then,
\begin{enumerate}[label=$(\roman*)$,leftmargin=2.5em]
\item Condition {\rm(F1)} is equivalent to {\rm(A1)};

\item Conditions {\rm(F1)--(F3)} imply {\rm(A1)} and {\rm(A3)};

\item Conditions {\rm(F1)--(F4)} imply {\rm(A1)--(A3)}.
\end{enumerate}
\end{proposition}
\begin{proof}
$(i)$  $\vphantom{l}$ Since $d=1$, we have $R_\eta=\sign(\eta)[0,\eta]$, and part $(i)$
directly follows from Theorem 25.17 in \cite{Sato}.
\par
$(ii)$ $\vphantom{l}$
We denote $f_{t,\eta}(x)\coloneqq \ee{\eta x}f_t(x)$, $f_{t,\eta}^{sym}(x)\coloneqq\big(f_{t,\eta}(x) + f_{t,\eta}(-x)\big)/2$ and $f_{t,\eta}^{asym}\coloneqq f_{t,\eta}^{sym} - f_{t,\eta}^{sym}$. Then the elementary identity $ab+cd=(a+c)(b+d)/2 + (a-c)(b-d)/2$ yields
\begin{align*}
f_{t,\eta}^{sym}(x) = \cosh(\eta x) f_{t,\eta}^{sym}(x) + \sinh(\eta x) f_{t,\eta}^{asym}(x).
\end{align*}
We notice that there exists a constant $c>0$ such that 
\begin{align*}
\cosh(\eta x) \ge \ee{-\eta \epsilon} \,\text{and }\big|\sinh(\eta x)\big|\le c |x|\quad\text{for every $|x|<\epsilon$.}
\end{align*}
Moreover, since $f\ge0$, the triangle inequality implies that $|f_{t,\eta}^{asym}(x)|\le f_{t,\eta}^{sym}(x)$ for every $(t,x)\in[0,T]\times \rr$. This shows that Condition (F3) also remains valid when we replace $f_{t}^{sym}$ by $f_{t,\eta}^{sym}$. Now part $(ii)$ follows from inequality (4.16) in the proof of Proposition 4.2 in \cite{Glau2013}.
\par
$(iii)$ $\vphantom{l}$ Along the same lines as in the proof of part $(ii)$, we observe that
\begin{align*}
f_{t,\eta}^{asym}(x) = \sinh(\eta x) f_{t,\eta}^{sym}(x) + \cosh(\eta x) f_{t,\eta}^{asym}(x).
\end{align*}
Thus, the validity of Condition (F4) also remains valid when replacing $f_{t}^{asym}$ by $f_{t,\eta}^{asym}$. Then, Proposition 4.2 in \cite{Glau2013} shows the assertion for $\alpha=1$.

For $\alpha<1$ we have $b_t= \int h(x) F_t(\dd x)$. According to Lemma \ref{lem-Aeta}, and using the notation therein, $b_t^{-\eta'} = \int h(x) F_t^{-\eta'}(\dd x)$. Hence, for all $t\in[0,T]$ and~$\eta'\in R_{\eta}$,
\begin{align*}
Im\big(A(\xi-i\eta')\big) = \int_{\rr} \sin(\xi x) \ee{-\eta' x}F_t(\dd x).
\end{align*}
Estimating the real part of the $A(\xi-i\eta')$ along the same lines as in the proof of Proposition 4.2 in \cite{Glau2013}, we obtain continuity condition (A2).
\qed
\end{proof}

We now apply Proposition \ref{prop-cond-F} to a concrete class of processes, which is frequently used to model asset prices: 
 
\begin{example}[Univariate generalized tempered stable L\'evy process]\label{bsp-A-temperedCGMY}
\textup{
Let $L$ be a generalized tempered stable L\'evy process with parameters $C_-$, $C _+\ge 0$ such that $C_-+C_+>0$ and $G$, $M>0$ and $Y_-, Y_+<2$. That is, $L$ is a pure jump L\'evy process with characteristic triplet $(b,0,F^{\operatorname{temp}};h)$ with $ F^{\operatorname{temp}}(\dd x) = f^{\operatorname{temp}}(x) \dd x$, where
\begin{eqnarray*}
f^{\operatorname{temp}}(x) = 
 \left\{
 \begin{array}{ll}
\frac{C_-}{|x|^{1+Y_-}}  \ee{G x}  \quad&\text{for }x<0 \\
\frac{C_+}{|x|^{1+Y_+}}  \ee{-Mx}  \quad&\text{for }x\ge0.
 \end{array}
\right.
\end{eqnarray*}
For $C_\pm=0$ we set $Y_\pm:=0$. 
\notiz{ Note that this definition coincides with the original definition in \cite{CarrGemanMadanYor2002} except for the choice of the truncation function resp. the drift parameter.}
For $C=C_-=C_+$ and $Y=Y_-=Y_+$ this class is known as CGMY, after Carr, Geman, Madan and Yor. Tempered stable  processes are also referred to as Koponen and KoBoL in the literature, see e.g.\ \cite{BoyarchenkoLevendorskii2002}. For the general setting see for instance \cite{PoirotTankov06}. 
}

By Proposition \ref{prop-cond-F}, Conditions (A1)--(A3) are satisfied for weight $\eta\in(-G,M)$ and Sobolev index $\alpha:=\max\{ Y_+,Y_-\}$ in each of the following cases:
\begin{enumerate}[label=$(\roman{*})$,leftmargin=3em]
\item  $\alpha=\max\{ Y_+,Y_-\}> 1$,
\item $Y:=Y_-=Y_+=1$ and $C_-=C_+$,
\item $0<\alpha=\max\{ Y_+,Y_-\}<1$ and $b = \int h(x)F(\dd x)$. 
\end{enumerate}
\end{example}

Further examples 
are considered in \cite{Glau2013} for the case $\eta=0$. There, Examples 4.5--4.7 give conditions on the parameters that imply Conditions (A2) and (A3) for 
generalized student-t, Cauchy, generalized hyperbolic and stable processes. Moreover, Section~4.3  in \cite{Glau2013} provides a sufficient tail condition on the L\'evy measure for (A2) and (A3) to hold.

\subsection{Time-inhomogeneous processes}\label{subsec-PIIAC}
When modeling with L\'evy processes in finance we often need to consider the larger class of time-inhomogeneous L\'evy processes, because their flexibility in time leads to a considerably better fit to the time-evolution of data. 
We therefore propose two construction principles that lead to parametric families of time-inhomogeneous processes satisfying (A1)--(A4).

First, we find it natural to define a family of time-inhomogeneous L\'evy processes by inserting time-dependent parameters into a given parametric class of L\'evy processes. For this class it turns out to be straightforward to show the following result.
\begin{lemma}\label{lem-At=Apt}
 Let $\OP\subset\rr^D$ and $(A(p,\cdot))_{p\in \OP}$ a parametrized family of symbols. Fix some $\eta\in\rrd$ and some $\alpha\in(0,2)$. Let (A2) and (A3) be satisfied for $A$, uniformly for all $p\in\OP$. 
Then, if $t\mapsto p(t)$ measurable, then (A2) and (A3) are satisfied for
\[
A_t(\xi):=A(p(t),\xi)\qquad\text{ for $t\in[0,T]$ and $\xi\in U_{-\eta}$.}
\]
If, moreover, $(p,\xi)\mapsto A(p,\xi)$ is continuous and $t\mapsto p(t)$ is is c\`adl\`ag, then $(A_t)_{t\ge0}$ is the symbol of a time-inhomogeneous L\'evy process $L'$ and also satisfies (A4). If additionally (A1) is satisfied for $L$, then it is also satisfied for~$L'$.
\end{lemma}

For $p(t)$ we can for instance choose a vector of piecewise constant parameters, so as to incorporate different short-, mid- and long-term behaviour.


As another natural construction let us consider stochastic integrals of deterministic functions with respect to L\'evy processes. Let $L$ be an $\rrd$-valued L\'evy process and $f$ a deterministic $L$-integrable $\rr^{n\times d}$-valued function. Then 
\[
X_t := f\cdot L_t:=\int_0^t f(s) \dd L_s := \Bigg( \sum_{k=1}^d\int_0^t f^{jk}(s) \dd L_s^k\Bigg)_{j\le d}
\]
defines an $\rr^n$-valued semimartingale with deterministic characteristics. Denote by $(b,c,F;h)$ the characteristics of $L$. Applying standard  arguments from the semimartingale theory, we see that the characteristics $(b^X_t,c^X_t,F^X_t;\tilde{h})_{t\ge0}$ of $X$ 
are given by 
\begin{align}
b^X_t&=f(t)b + \int_{\rrd}\big(\tilde{h}(f(t)x) - f(t) h(x) \big)F(\dd x),\nonumber\\
c^X_t&=f(t)cf(t)^{tr},\label{charX}\\
F^X_t(B) &= \int_{\rrd} \1_{B}\big(f(t)x\big) F(\dd x) \quad\text{for every }B\in \OB\big(\rrd\setminus\{0\}\big).\nonumber
\end{align}  
In particular, $X$ is a time-inhomogeneous L\'evy process in the sense of our definition provided that integrability condition \eqref{int-PIIAC} is satisfied for its characteristics. Moreover, if $A$ is the symbol of $L$,
the symbol $A^X$ of $X$ is given\tild by 
\begin{align}\label{At=Aft}
A_t^X(\xi)= A\big( f(t)^{tr}\xi \big)+ i\skl \xi ,  b(\tilde{h},h,f)\skr \quad\text{for every }\xi\in\rrd,
\end{align}
where $b(\tilde{h},h,f):= \int_{\rrd}\big(\tilde{h}(f(t)x) - f(t) h(x) \big)F(\dd x)$.  This generalizes Example\tild 7.6 in \cite{EberleinGlau2013}, where $f:[0,\infty)\to\rr_+$.
\begin{lemma}\label{lem-At}
Let $L$ be a L\'evy process that is also a special semimartingale and let\tild $A$ denote its symbol. Let $f:[0,\infty)\to\rr^{n\times d}$ be a measurable function such that there exist constants $0<f_\ast,f^\ast$ with
\begin{align}\label{ftabsch}
\sup_{0\le t\le T}\|[f(t)f(t)^{tr}]^{-1}\|^{1/2}\le f_\ast^{-1} \,\text{and }\sup_{0\le t\le T}\|f(t)f(t)^{tr}\|^{1/2}\le f^\ast,
\end{align}
where $\|\cdot\|$ denotes the spectral norm. Then $X:=f\cdot L$ is a time-inhomogeneous L\'evy process as well as a special semimartingale  and its symbol is given by 
\[
A^X_t(\xi) = A\big(f(t)^{tr}\xi\big)\quad\text{ for all }\xi\in\rr^n.
\]
Fix some $\rho>0$, $\eta^X\in\rrd$ with $|\eta^X|\le \frac{\rho}{f^\ast}$ and some $\alpha>0$.
If $E\ee{\rho |L_t|}<\infty$ for some $t>0$, then $X$ satisfies $(EM(R_{-\eta^X}))$. If additionally $A$ satisfies (A2) and (A3) for every weight $\eta\in\rrd$ with $|\eta|\le\rho$ and index $\alpha>0$, then (A2) and (A3) hold for $A^X$ with the same index $\alpha$ and weight~$\eta^X$. Moreover, if (A4) holds for $A$ it is also satisfied by $A^X$.
\end{lemma}
\begin{proof}
From the assumptions it is immediate that $f$ is integrable with respect to $L$ and, hence, $X$ is a semimartingale with characteristics of form\tild \eqref{charX}. As integrability condition \eqref{int-PIIAC} also follows directly, we see that $X$ is a time-inhomogeneous L\'evy process.
Since $L$ is a special semimartingale, we have $\int_{|x|>1}|x|F(\dd x)<\infty$, where $F$ denotes the L\'evy measure of $L$, and \eqref{ftabsch} implies
\begin{align}\label{sepecials}
\int_{0}^T\int_{|x|>1}|x|F_t(\dd x)
\le T f^\ast\int_{|x|>1/f_\ast}|x|F(\dd x)
<\infty.
\end{align}
This shows that also $X$ is a special semimartingale. Therefore we may choose both $h$ and $\tilde{h}$ as the identity so that $b(\tilde{h},h,f)=0$. From \eqref{At=Aft} we now obtain the equality $A^X_t(\xi) = A\big(f(t)^{tr}\xi\big)$. The assertion as to the exponential moment condition~(A1) we obtain analogously to \eqref{sepecials}. The assertions on (A2)--(A4) follow immediately from 
the continuity of L\'evy symbols and Lemma\tild \ref{lem-Aeta}.
\qed
\end{proof}

\section{Applications}\label{sec-applications}

Having convinced ourselves that it is a wide and interesting class of stochastic processes for which Theorem \ref{fkac} links conditional expectations with weak solutions of PIDEs, let us now explore the virtues of the result for applications.
Starting with pricing problems in finance, where discontinuous killing rates arise naturally,
 we furthermore find that indicator type killing rates also help us to characterize interesting probabilistic objects. In all of these applications the driving process $L$ can be chosen freely and we may employ jump-diffusions or pure jump processes. The latter are intensely used in finance. Examples are NIG and generalized tempered stable processes, which we have shown satisfy the assumptions of Theorem~\ref{fkac}.
Finally we encounter the  original ideas of Feynman and Kac in a relativistic guise.
In the context of the relativistic Schr\"odinger equation we shall see the family of NIG processes in a fundamental role.

\subsection{Employee options}\label{sec-empop}
We propose a class of employee options that flexibly reward the management board according to the performance of the corporation's stock price. 
Financial instruments used in this context are called \textit{employee stock options} and often are based on European call options. Thus the reward depends on the level of the stock at specific points in time. 
Shareholders though typically are interested in the performance of the stock during the whole period. 
They mean to support management decisions that push the stock price constantly to a high level. 
Moreover, it is arguably fairer to choose the reward according to the performance of the stock value as relative to the market evolution. 

To make this formally precise,
denote by $S$ the $d$-dimensional stochastic process that models the stock of the company and $d-1$ reference assets. Let $G:\rrd\to\rr$ be a payout profile and $\kappa:[0,T]\times\rrd\to \rr$ a \emph{reward rate function}. For $\kappa<0$ the reward turns into a penalty.
Moreover, we include a continuously paid salary   by the \emph{salary function} $f:[0,T]\times\rrd\to\rr$. 
At maturity $T$ the employee obtains the payout
\begin{equation}\label{Gofoption}
G(S_T)\ee{\int_0^T  \kappa_h(S_{h})\dd h},
\end{equation}
in addition to the salary
\begin{equation}\label{fofoption}
f(t,S_t)\ee{\int_0^t\kappa_h(S_{h})\dd h}\dd t,
\end{equation}
which is paid at each instant $t\in[0,T]$. Thus, the payout profile $G$ may depend on the level of the stock and the reference assets. The reward rate and the salary may additionally be time-dependent. Note that our analysis allows us to incorporate \textit{discontinuities} in the reward rate. Thus threshold and indicator type reward functions are allowed, which is a natural choice. Indicator type killing rates for instance play the role of instantaneous rewards or penalties for stock price levels in a specified domain.

We further use the following notation.  For $x=(x_1,\ldots,x_d)\in\rrd$, let $\ee{ x} :=(\ee{x_1},\ldots,\ee{x_d})$, $\widetilde{G}(x):=G(\ee{x})$, $\tilde{\kappa}(\cdot,x):= - \kappa(\cdot,\ee{x})$ and $\tilde{f}(\cdot,x):=f(T- \cdot,\ee{x})$.
We assume the interest rate $(r_t)_{t\ge0}$ to be deterministic, measurable and bounded. We model $S=(S_0^1\ee{L^1},\ldots,S_0^d\ee{L^d})$ by a time-inhomogeneous L\'evy process  $L$ with local characteristics $(b,c,F;h)$ such that the no-arbitrage condition, 
\begin{align}
\label{drift-cond}
b^i_t = r_t -  \frac{1}{2} c^{ii}_{t}- \int (e^{x_i}-1- h_i(x)) F_{t}(\dd x) \quad\text{for every $i=1,\ldots,d$},
\end{align}
is satisfied, 
where $h_i$ is the $i$-th component of the truncation function $h$.

The following assertion shows that the fair price of the employee option specified by~\eqref{Gofoption} and \eqref{fofoption} can be computed by solving the related Kolmogorov PIDE. The result is an immediate consequence of Theorem \ref{fkac}.
\begin{corollary}\label{cor-employee}
Let $\eta\in\rrd$ and $\alpha\in(0,2]$ such that $\widetilde{G} \in L^2_\eta(\rrd)$ and assume the time-inhomogeneous L\'evy process $L$ satisfies \eqref{drift-cond} and Conditions (A1)--(A4). Denoting $x=\log(S_0)$, 
 the  fair price
\[
 u(T,x):=E_x\Big(\widetilde{G}(L_T)\ee{ \int_0^T(\tilde{\kappa}_h(L_h)-r_h)\dd h} + \int_0^T \tilde{f}(T-s,L_s) \ee{\int_0^s(\tilde{\kappa}_h(L_h)-r_h)\dd h}\dd s\Big)
\]
of the employee option with payout profile~\eqref{Gofoption}, \eqref{fofoption} 
 is given by the unique weak solution  $u\in W^1\big(0,T ;H^{\alpha/2}_\eta(\rr^d), L^2_\eta(\rrd)\big)$ of
\begin{align}\label{para-eqempl}
\dot{ u} + \OA_{T-t}u + \tilde{\kappa}_{T-t} u &= -\tilde{f}, \quad
u(0) =\, \widetilde{G}\,.
\end{align}
\end{corollary}
See Section~\ref{sec-num} for a numerical implementation of equation \eqref{para-eqempl}.

\subsection{L\'evy-driven short rate models}
L\'evy driven term structure models were introduced first in \cite{EberleinRaible99}. 
Here, we consider a short rate of the form
\begin{equation}\label{r}
r_t \coloneqq r(t, L_t)
\end{equation}
with an $ \rrd$-valued time-inhomogeneous L\'evy process $L$ and a measurable and bounded interest rate function $r:[0,T]\times \rr^d \to \rr$. We allow for discontinuities in the function $r$ and thus for thresholds in factor model \eqref{r}.

At maturity, the holder of a \emph{zero coupon bond} receives one unit of currency. In accordance with the no-arbitrage principle, the time-$t$ value of the zero-coupon bond with maturity $0\le t\le T$ is modeled by
\begin{equation}\label{bond}
P(t,T)\coloneqq E\big( \ee{-\int_t^T r_h\dd h}\,\big|\OF_t\big).
\end{equation}
Translating this conditional expectation formally into an evolution problem of form \eqref{parabolic-eq-origin}, we obtain $g(x)\equiv 1$ as initial condition.
We now have to realize that there is no weight $\eta\in\rrd$ such that $x\mapsto\ee{\skl\eta,x\skr}\in L^2(\rrd)$. We therefore split the initial condition into summands that each lie in a weighted $L^2$-space. 
In the one-dimensional case, for example, we have 
$g=\1_{(-\infty,0]} + \1_{(0,\infty)}$, where $1_{(-\infty,0]}\in L^2_{\eta^-}( \rr)$ for every $\eta^->0$ and $1_{(0,\infty)}\in L^2_{\eta^+}( \rr)$ for every $\eta^+<0$.

\begin{remark}\label{rem-split-orthants}
We split the initial condition $g$ into in $2^d$ summands $g^j$ that are supported in the $2^d$ orthants. 
To be precise,  for $j=1,\ldots,2^d$, let $p^j:=(p^j_1,\ldots,p^j_d)$ with $p^j_i\in\{-1,1\}$ for the $2^d$ different possible configurations and let 
\begin{equation}\label{def-Oj}
O^j\coloneqq \big\{(x_1,\ldots,x_d)\in\rrd \,\big|\, p^j_ix_i\ge 0\,\text{for all }i=1,\ldots,d\big\}.
\end{equation}
By linearity of expectation, respectively of the PIDE, the problem can be split additively in $2^d$ separate problems. If for each of the summands $g^j$ a weight $\eta^j\in\rrd$ exists such that  $g^j\in L^2_{\eta^j}(\rrd)$, then the results of Theorem \ref{fkac} can be applied to each problem with  initial condition $g^j$ separately.
 
\end{remark}
As in Remark \ref{rem-split-orthants} we split the unity in the following way: $1 \equiv g(x)=\sum_{j=1}^{2^d}\1_{O^j}(x)$ a.e. with the distinct orthants $O^j$ of $\rrd$ given by \eqref{def-Oj}.
For each $j$, we choose 
\begin{equation}\label{def-etaj}
\eta^j:= -\epsilon d^{-1/2}p^j
\end{equation}
so
that $\1_{O^j}\ee{\skl\eta^j,\cdot\skr}\in L^2( \rrd)$. 
If the distribution of $L_T$ has a Lebesgue density, we may rewrite equation~\eqref{bond} as
\begin{equation}\label{eqsplitbond}
\scalebox{0.94}[1]{$\displaystyle  u(T-t,x)
=
\sum_{j=1}^{2^d} u^j(T-t,x)\,\,\text{with } u^j(T-t,x):=
E_x\big(\1_{O^j}(L_T)\ee{-\int_t^T \! r_h\dd h}\big). $}
\end{equation}
%

%

%
\begin{corollary}\label{prop-bond}
For $\epsilon>0$ and $\alpha\in (0,2]$, 
let $L$ be a time-inhomogeneous L\'evy process such that $E\ee{\epsilon|L_t|}<\infty$ for every $t\le T$  and its symbol $A$ satisfies (A2)--(A4) for index $\alpha$ and every weight $\eta\in\rrd$ with $|\eta|<\epsilon$.
Then for every $0\le t<T$, the price of the zero coupon bond in model \eqref{r} is given as 
\[
P(t,T)= \sum_{j=1}^{2^d} u^j(T-t,L_t)
\quad\text{a.s.,}
\]
where $u^j$ is the unique weak solution in $W^1\big(0,T ;H^{\alpha/2}_{\eta^j}(\rr^d), L^2_{\eta^j}(\rrd)\big)$ of
\begin{align}\label{para-eqj2}
\dot{ u}^j + \OA_{T-t}u^j + r u^j &= 0, \quad
u(0) =\, \1_{O^j}.
\end{align}
\end{corollary}
\begin{proof}
The assumptions yield that
for each $j=1,\ldots, 2^d$, Conditions (A1)--(A4) are satisfied for weight $\eta^j$ and index $\alpha$. 
According to Remark \ref{rem-implicationsA1A3}, the distribution of $L_T$ has a Lebesgue density, which yields equation \eqref{eqsplitbond}. 
Now, the assertion follows directly from Theorem \ref{fkac}.
\qed
\end{proof}
It is worth mentioning that with the same technique we can characterize prices of options on a zero-coupon bond by solutions of PIDEs. A distinctive feature of the resulting PIDE is that the solution $u$ of equation \eqref{eqsplitbond} appears as the initial condition.  Its initial condition thus is given by the solutions to  PIDEs\tild \eqref{para-eqj2}.

Interesting related applications are bankruptcy probabilities in the model of \cite{AlbrecherGerberShiu2011}, the value of barrier strategies in the bankruptcy model of \cite{AlbrecherLautscham2013} and reduced form modelling of credit risk as in \cite{JeanblancLeCam2007}.




\subsection{Laplace transform of occupation times}
\label{sec-oc}
We characterize Laplace transforms of occupation times of time-inhomogeneous L\'evy processes via weak solutions of PIDEs.
Setting $\kappa\coloneqq \1_D$ for $D\subset\rrd$ Borel measurable, $f\coloneqq 0$, initial condition $g\coloneqq 1$, and inserting $L_0=x$, equation \eqref{gl-stochdarst} from Theorem \ref{fkac} becomes
\begin{equation}\label{eqLaplace1}
 u(T,x)
=
E_x\big(\ee{-\gamma \int_0^T   \1_D(L_{h}) \dd h}\big),
\end{equation}
which is the Laplace transform at $\gamma$ of the \textit{occupation time} $\int_0^T   \1_D(L_{h}) \dd h$ that the process $L$ spends in the domain $D$ until time $T$. Landriault, Renaud and Zhou (2011)\nocite{LandriaultRenaudZhou2011} analyse Laplace transforms of occupation times of spectrally negative L\'evy processes using fluctuation identities. We characterize these transforms for a wide class of time-inhomogeneous L\'evy processes by parabolic PIDEs. Let us point out that the assertion is not restricted to spectrally negative processes as the examples of NIG and tempered stable processes show, see Section \ref{sec-ex}.

Splitting the corresponding initial condition according to Remark \ref{rem-split-orthants}, we let
%
\begin{equation}\label{eqLaplace1eq}
u^j(T,x):=
E_x\big(\1_{O^j}(L_T)\ee{-\gamma \int_0^T   \1_D(L_{h}) \dd h}\big). 
\end{equation}
%
%
%
Arguing as in the  proof of Corollary \ref{prop-bond} and applying Corollary~\ref{cor-hoelder}, we obtain: 
\begin{corollary}\label{prop-Laplace}
For $\epsilon>0$ and  $\alpha\in(0,2]$, 
let $L$ be a time-inhomogeneous L\'evy process such that $E\ee{\epsilon|L_t|}<\infty$ for every $t\le T$  and its symbol $A$ satisfies (A2)--(A4) for index $\alpha$ and every weight $\eta\in\rrd$ with $|\eta|<\epsilon$. Let $\eta^j:= -\epsilon d^{-1/2}p^j$ as in \eqref{def-etaj}.
 Then, $u^j$ from equation \eqref{eqLaplace1eq} is the unique weak solution in the space $W^1\big(0,T ;H^{\alpha/2}_{\eta^j}(\rr^d), L^2_{\eta^j}(\rrd)\big)$~of
\begin{align}\label{para-eqj}
\dot{ u}^j + \OA_{T-t}u^j + \1_{D} u^j &= 0, \quad
u(0) =\, \1_{O^j}
\end{align}
and $u$ from equation \eqref{eqLaplace1} is given by
$$u(T,x)=\sum_{j=1}^{2^d} u^j(T,x).$$
If $d=1$ and $\alpha>1$, then $x\mapsto u(t,x):=E_x\big(\ee{-\gamma \int_0^t   \1_D(L_{h}) \dd h}\big)$ is $\lambda$-H\"older continuous with $\lambda=\frac{\alpha-1}{2}$ for each $t\in [0,T]$ and in particular also continuous.
\end{corollary}
%

\subsection{Penalization of the domain}\label{sec-pendomain}
Observe that the limit of $\ee{-\gamma \int_0^T   \1_{\overline{D}^c}(L_{h}) \dd h}$ as $\gamma\to\infty$ links occupation times to exit times.
This idea lies at the basis of the repeated use of occupation times for modeling. Moreover, it opens a way to establish a Feynman-Kac-type representation of type \eqref{gl-stochdarst} for boundary value problems. In the language of diffusions, the presence of particles in the outer domain is penalized stronger and stronger until it is finally killed the moment it leaves the domain.
 For jump diffusion processes, the argument is outlined in \cite{BensoussanLions}. 
  In \cite{PhdGlau} and a forthcoming article, \cite{GlauFkac2015}, a similar technique is used for time-inhomogeneous L\'evy processes. Interesting for finance, the resulting Feynman-Kac-type representation  serves to characterize prices of barrier and lookback options in pure jump models. The argument is based on the following result and the convergence of solutions for a sequence of killing rates of indicator type given by $\kappa^\lambda(x):=\lambda\1_{\overline{D}^c}(x)$ for $\lambda\to \infty$.
\begin{corollary}\label{prop-penalization}
For $\alpha\in(0,2]$ and $\eta\in\rrd$,
let $L$ be a time-inhomogeneous L\'evy process satisfying assumptions (A1)--(A4).
Let $f\in L^2\big(0,T ; H^l_\eta(\rrd)\big)$ for some $l\ge 0$ with $l>(d-\alpha)/2$, $g\in L^2_\eta(\rrd)$,  $\kappa:[0,T]\times \rrd \to\rr$ measurable and bounded, $\lambda>0$ and $D\subset\rrd$ open. Then the unique weak solution $u^\lambda\in W^1\big(0,T;H^{\alpha/2}_\eta(\rrd),L^2_\eta(\rrd)\big)$ of
\begin{align}\label{parabolic-eq-lambda}
\begin{split}
\partial_t u^\lambda + \OA_{T-t} u^\lambda + \kappa_{T-t} u^\lambda + \lambda \q u^\lambda  =&\, f, \quad
u^\lambda(0) =\, g\,,
\end{split}
\end{align}
has for every $t\in(0,T]$  almost surely the stochastic representation
\begin{equation}\label{gl-stochdarstlambda}
 \begin{split}
\scalebox{.93}[1]{$\displaystyle  u^\lambda(T-t,L_t)$}
&=
E\Big(g(L_T)\ee{-\int_t^T   \kappa_h(L_{h-}) \dd h} \ee{- \lambda\int_t^T  \q(L_{h-}) \dd h} \\
&
\quad \scalebox{.97}[1]{$\displaystyle+ \! \int_t^{T}\! \!\ f(T-\!s,L_s)\ee{-\int_t^s   \kappa_h(L_{h-}) \dd h} \ee{- \lambda\int_t^s \!\q(L_{h-}) \dd h}\dd s \Big|\OF_t\Big). $}
 \end{split}
\end{equation}
\end{corollary}
\begin{proof}
The assertion follows directly from Theorem \ref{fkac}.
\qed
\end{proof}

\subsection{Relativistic Schr\"odinger equation}\label{sec-RelSchr}
Our analysis leads us  
back to the origin of Feynman and Kac's deep link between Schr\"odinger's equation and diffusion processes.
Recast in a relativistic mold, the formalism
brings
Normal Inverse Gaussian  L\'evy processes into the spotlight: 
%
%
%
%
We find that a specific NIG process plays the same role for the relativistic Schr\"odinger equation as the Brownian motion does for the classical Schr\"odinger equation.
Carmona, Masters and Simon (1990)\nocite{CarmonaMastersSimon1990} provide a Feynman-Kac-type formulation of this link but give no formal proof. 
Baeumer, Meerschaert and Naber (2010)\nocite{BaeumerMeerschaertNaber2010} exploit this relation to model relativistic particle diffusion by an NIG process. We follow their presentation of the connection between the relativistic Schr\"odinger equation and NIG processes. Then, Theorem \ref{fkac} allows us to make this link formally precise.

The nonrelativistic \emph{Schr\"odinger equation} for a single particle in a quantum system described by the \emph{potential energy} $V:\rrd\times \rr_+\to \rr$ is the following partial differential equation for the \emph{wave-function} $\psi:\rrd\times \rr_+\to \cc$,
\begin{equation}\label{Schroedinger}
\mathrm i\hbar \frac{\partial\psi}{\partial t}(x,t) \;=\; \Big(- \frac{\hbar^2}{2m}\Delta + V(x,t)\Big)\psi(x,t),
\end{equation}
where $i$ is the imaginary unit, $\frac{\partial\psi}{\partial t}$ denotes the time derivative of $\psi$, $2\pi\hbar$ is  \emph{Planck's constant}, $m$ is the \emph{particle's mass}, and the \emph{Laplace operator} $\Delta$ is given by $\Delta \psi(x,t)\coloneqq 
\sum_{j=1}^d\frac{\partial^2\psi}{\partial x^2_j}(x,t)$.

For a free particle, i.e. if $V\equiv 0$, a formal connection to the Kolmogorov backward equation of the Brownian motion is obtained by the analytic continuation of the Schr\"odinger equation \eqref{Schroedinger} in time and inserting $\tau=it$. For $V\not\equiv 0$, setting $V(x,it):= V(x,t)$ for every $x$ and $t$, this relates equation\tild \eqref{Schroedinger}\tild to 
\begin{equation}\label{Schroedinger-Kolmogorov}
\hbar \frac{\partial\psi}{\partial t}(x,t) \;=\; \Big(\frac{\hbar^2}{2m}\Delta - V(x,t)\Big)\psi(x,t),
\end{equation}
which is the Kolmogorov backward equation of the killed Brownian motion with volatility $\sigma=\sqrt{\hbar/2m}$ and killing rate $V/\hbar$.

Let us now pass to the \emph{relativistic Schr\"odinger equation}. According to Baeumer, Meerschaert and Naber (2010)\nocite{BaeumerMeerschaertNaber2010}, the \emph{relativistic kinetic energy} of a particle with rest mass $m$ and momentum $p$ is given by
\begin{equation}\label{relenergy}
E(p)=\sqrt{\|p\|^2c^2 + m^2 c^4}-m c^2,
\end{equation}
where $c$ denotes the speed of light. The relativistic energy \eqref{relenergy} serves as a pseudo differential operator to define the relativistic Schr\"odinger operator 
\begin{equation}
\mathcal{H}_0(\psi)(\cdot,t)\coloneqq \OF^{-1} (E\OF(\psi(\cdot,t)))
\end{equation}
for the free particle. Thus, the relativistic Schr\"odinger equation for a single particle in a quantum system described by the potential energy $V$ is given by
\begin{equation}\label{relSchroedinger}
\mathrm i\hbar \frac{\partial\psi}{\partial t}(x,t) \;=\; \big(\mathcal{H}_0 + V(x,t)\big)\psi(x,t).
\end{equation}
Analogous to the nonrelativistic case, formally inserting $\tau =it$ in equation~\eqref{relSchroedinger} and setting $V(x,it):= V(x,t)$ for every $x$ and $t$, yields
\begin{equation}\label{reellrelSchroedinger}
 \frac{\partial\psi}{\partial t}(x,t) +\frac{1}{\hbar}\big(\mathcal{H}_0 + V(x,t)\big)\psi(x,t) =0.
\end{equation}
We note that $\frac{1}{\hbar}E(p)$ is the symbol of the NIG process $L$ with parameters $\tilde{\alpha}=mc^2$, $\beta=0$, $\delta=\frac{1}{\hbar}$, $\mu=0$ and $\Delta=c^2\operatorname{Id}_{d}$, where we use the notation of Example~\ref{ex-multidNIG} and $\operatorname{Id}_{d}$ denotes the identity matrix in $\rrd\times\rrd$.

The following corollary formally establishes the Feynman-Kac-type relation of equation~\eqref{reellrelSchroedinger} to NIG processes in terms of weak solutions. Note that here the potential energy $V$ is allowed to be discontinuous.
\begin{corollary}\label{prop-relSchroedinger}
Let the potential energy $V$ be measurable and bounded.
Let $g\in L^2_\eta(\rrd)$ for some $\eta\in\rrd$ such that $\|\eta\|^2\le  m^2c^2$. Then the unique weak solution $u\in W^1\big(0,T;H^{1/2}_\eta(\rrd),L^2_\eta(\rrd)\big)$ of
\begin{align}\label{parabolic-schroedinger}
\begin{split}
\dot{u} + \frac{1}{\hbar}(\mathcal{H}_0 u + V u)  =&\, 0, \quad
u(0) =\, g\,,
\end{split}
\end{align}
has for every $t\in(0,T]$ the stochastic representation
\begin{equation}\label{gl-stochdarst-schroedinger}
 \begin{split}
  u(T-t,L_t)
&=
E\Big(g(L_T)\ee{-\frac{1}{\hbar} \int_t^T   V_{T-h}(L_{h}) \dd h} \,\Big|\,\OF_t\Big)\quad\text{a.s.}
 \end{split}
\end{equation}
\end{corollary}
\begin{proof}
Corollary \ref{prop-relSchroedinger} is a direct consequence of Theorem \ref{fkac} and Example \ref{ex-multidNIG}.
\qed
\end{proof}
%

\section{Numerical implementation}\label{sec-num}
Let us now explore the practical benefits of our Feynman-Kac-type result. 
We therefore implement a numerical scheme to solve Kolmogorov equation \eqref{parabolic-eq-origin} 
for pricing path dependent options in jump models. We specify a class of employee options so as to shed light on 
the effect of a discontinuous killing rate.
In order to give insight in the technique,
as numerical scheme we choose the wavelet Galerkin method as developed by \cite{MatachePetersdorffSchwab2004} for European option pricing. This is a very powerful method, which uses compression techniques and can be adapted to more involved pricing problems, 
as we will demonstrate by incorporating killing rates. The implementation requires some results from the classical theory on numerical analysis on partial differential equations. In addition, the jump part of the operator needs some special treatment. We take care of the derivation of the discrete scheme  by presenting the discretization steps (1)--(6) below.

We specify a type of employee option as described in Section~\ref{sec-empop}.
So as to include penalizations of low stock values permanently rather than only at a fixed maturity, we combine a call option with an indicator type killing rate. We specify the latter as instantaneous penalization for stock values below a fixed level by setting $\kappa(S):=-\lambda \1_{(-\infty, B]}(S)$ with a scale factor $\lambda>0$ and level $B$ in equation \eqref{Gofoption}, i.e.\ the payout at maturity is given by
\begin{equation*}
G(S_T)\ee{ - \int_0^T \lambda \1_{(-\infty, B]} (S_{h})\dd h},
\end{equation*}
where $G(S):=\max\big\{S-K,0\big\}$.
As driving process $L$ in the model $S=S_0\ee{L}$ we choose a pure jump L\'evy process from the family of CGMY processes described in Example \ref{bsp-A-temperedCGMY} with parameters $C>0$, $G>1$, $M>0$, $Y\in[1,2)$, and whose drift $b$ is given by the no-arbitrage condition \eqref{drift-cond}. Then, according to Example \ref{bsp-A-temperedCGMY}, the process and its symbol satisfy Conditions (A1)--(A4), with weight  $\eta\in(G,-1)$ and index $\alpha=Y$.

We now fix a weight $\eta\in(G,-1)$, denote by $\OA$ the Kolmogorov operator of the process and let $\widetilde{G}(x):=G(\ee{x})$ and $\tilde{\kappa}(x):= - \kappa(\ee{x})$.
According to Corollary\tild \ref{cor-employee} we obtain the fair price of the option by computing the unique weak solution  $u\in W^1\big(0,T ;H^{Y/2}_\eta(\rr), L^2_\eta(\rr)\big)$ of
\begin{align}\label{eq-exact}
\dot{ u} + \OA u + (r+\tilde{\kappa}) u &= 0, \quad
u(0) =\, \widetilde{G}.
\end{align}



In order to prepare the discretization with finite elements, we first modify and then localize the equation to a bounded interval. The variational formulation of the resulting equation then allows us to discretize the space with a Galerkin method. Finally, the time discretization completes the fully discrete scheme. 
In more detail, we proceed along the following steps: 

%

\begin{enumerate}
\item[(1)] \emph{Modification of the equation:}

Choose a function $\psi\in W^1\big(0,T ;H^{Y/2}_\eta(\rr), L^2_\eta(\rr)\big)$ such that $\phi:=(u -\psi)\in W^1\big(0,T ;H^{Y/2}(\rr), L^2(\rr)\big)$ and $|\phi(t,x)|\to 0$ for $|x|\to \infty$. Then $\phi$ is the unique weak solution of the \emph{modified equation}
\begin{align}\label{eq-modified}
\dot{\phi} + \OA \phi + (r+\tilde{\kappa}) \phi &=\, f, \quad
\phi(0) =\,\widetilde{G} - \psi(0).
\end{align}


\item[(2)] \emph{Truncation to a bounded domain:}

We localize the equation to a bounded interval $(R_1,R_2)$ with zero constraints outside of the interval. Here, we for the first time encounter a conceptual difference between jump and non-jump processes:
The jump part of the process renders the operator $\OA$ nonlocal. It does therefore not suffice to specify zero boundary conditions. Rather, the values have to be set on the whole outer domain $\rr\setminus (R_1,R_2)$. Formally, we incorporate these zero constraints by defining the solution space as $\widetilde{H}^{Y/2}(R_1,R_2)\coloneqq \big\{ u\in H^{Y/2}(\rr)\big| u|_{[R_1,R_2]^c}=0\big\}$ and $\widetilde{L}^2(R_1,R_2)\coloneqq \big\{ g\in L^2(\rr)\big| g|_{[R_1,R_2]^c}=0\big\}$. To be precise, instead of solving equation \eqref{eq-modified} we approximate the unique weak solution $\tilde{\phi}\in W^1\big(0,T ;\widetilde{H}^{Y/2}(R_1,R_2) , \widetilde{L}^2(R_1,R_2)\big)$~of 
\begin{align}\label{eq-trunc}
\dot{\tilde{\phi}} + \OA \tilde{\phi} + (r+\tilde{\kappa}) \tilde{\phi} &=\, f, \quad
\tilde{\phi}(0) =\,\big(\widetilde{G} - \psi(0)\big)\1_{(R_1,R_2)}.
\end{align}
We now have to realize that we have changed the problem and that we need to control the resulting error $\|\phi - \tilde{\phi}\|$ with an appropriate norm $\|\cdot\|$. Put differently, we have to choose the function $\psi$ in step (1) in such a way that the error $\|\phi - \tilde{\phi}\|$ decays fast as $-R_1,R_2\to\infty$.\\[-1ex]
\item[(3)] \emph{Variational formulation of the equation:}

Weak solution
$\tilde{\phi}\in W^1\big(0,T ;\widetilde{H}^{Y/2}(R_1,R_2) , \widetilde{L}^2(R_1,R_2)\big)$ solves operator
equation \eqref{eq-trunc} if and only if $\tilde{\phi}$ satisfies the initial condition of \eqref{eq-trunc} as a limit in $\widetilde{L}^2$, that is
\begin{equation*}
 \lim_{t\rightarrow0} \tilde{\phi}(t)= \big(\widetilde{G} - \psi(0)\1_{(R_1,R_2)} \big)\quad \text{in } \widetilde{L}^2(R_1,R_2)
\end{equation*}
and for all $\nu \in C^\infty_0(0,T)$ and $\varphi \in \widetilde{H}^{Y/2}(R_1,R_2)$,
\begin{equation}\label{eq-variation1}
- \int_0^T \skl \tilde{\phi}(t), \varphi \skr_{L^{2}} \,\dot{\nu}(t) \dd t +  \int_0^T a( \tilde{\phi}(t), \varphi ) \, \nu (t) \dd t = \bar{f}(\varphi,\nu),
\end{equation}
with bilinear form $a:\widetilde{H}^{Y/2}(R_1,R_2)\times\widetilde{H}^{Y/2}(R_1,R_2)\to \rr$ and $\bar{f}(\varphi,\nu)\coloneqq \int_0^T \skl f(t),\varphi\skr_{L^2} \, \nu(t) \dd t$. For the simplicity of presentation, we assume from now on that $\psi$ is constant in time.
\\[-1ex]
\item[(4)] \emph{Space discretization with a Galerkin method:} 

Coming to the heart of the Galerkin method,
we choose a countable Riesz basis $\{w_1,w_2,\ldots\}$ of $\widetilde{H}^{Y/2}(R_1,R_2)$ and define
\begin{equation*}
 X_n:=\operatorname{span}\{w_1,\ldots,w_n\}\qquad\text{for all }n\in\N.
\end{equation*}
Since $\widetilde{H}^{Y/2}(R_1,R_2)$ is dense in $\widetilde{L}^{2}(R_1,R_2)$, we may further choose $h_n$ in\tild $X_n$ such that 
$
h_n \rightarrow \phi_0
$ in $\widetilde{L}^{2}(R_1,R_2)$.
We obtain the \textit{Galerkin equations} for each fixed $n\in\nn$ simply by restricting the variational equation \eqref{eq-variation1}. The resulting problem is:  
\textit{
Find a function $v_n \in  W^1\big(0,T; X_n; \widetilde{L}^2(R_1,R_2)\cap X_n \big)$ that satisfies for all $\chi \in C^\infty_0(0,T)$ and $\varphi \in X_n$,
\begin{equation}\label{gl-variation_V_n}
\begin{split}
-\int_0^T\!\! \skl v_n(t), \varphi\skr_{L^{2}} \,\dot{\chi}(t) \dd t + \int_0^T\!\! a\big(v_n(t), \varphi\big) \,\chi(t) \dd t 
&=
\bar{f}(\varphi,\chi)
 \\
v_n(0) &= h_n.
\end{split}
\end{equation}
}

Elegantly, the classical theory guarantees the convergence of the sequence\tild $v_n$ to $\tilde{\phi}$ already in this abstract setting. For more details we refer to Theorem\tild 23.A. and Remark 23.25 in \cite{Zeidler}.

The actual performance of the scheme, though, critically depends on the choice of the Riesz basis, which determines the rate of convergence.
\\[-1ex]

\item[(5)] \emph{Matrix formulation of equation \eqref{gl-variation_V_n}:}

Thanks to the linearity of the operators, we can simplify equation \eqref{gl-variation_V_n}. Namely, it is enough to insert the basis functions $w_1,\ldots, w_n$ as test functions $\varphi \in X_n$ in equation\tild \eqref{gl-variation_V_n}. Then, denoting $h_n \coloneqq \sum_{k=1}^n \alpha_{k} w_k$
and $v_n(t)\coloneqq \sum_{k=1}^n V_k (t) w_k$, equation \eqref{gl-variation_V_n} turns out to be equivalent to
\begin{align*}
\sum_{k=1}^n \dot{V}_{k}(t) \skl w_k,w_j\skr_{L^{2}} + \sum_{k=1}^n V_{k}(t) a\big(w_k,w_j \big) 
&= 
-
a\big(\psi,w_j \big)  \\
V_k(0)&= \alpha_k\quad\text{for all }k=1,\ldots,n.
\end{align*}
Written in matrix form the problem is to find $V:[0,T] \rightarrow \rr^{n}$ such that
\begin{align}\label{Pide-matrixform}
M \dot{V}(t) + A V(t) &= F \\
V(0) &=\alpha,\label{Pide-matrixform-initial}
\end{align}
where $F = (F_1, \ldots, F_n)^\trans$ with $F_k = -a\big(\psi,w_k \big)$ for $k=1,\ldots,n$, $\alpha =(\alpha_1,\ldots,\alpha_n)^\trans$, and the \emph{mass matrix} $M$ and \textit{stiffness matrix} $A$ are given\tild by
\begin{equation}
M_{jk} =  \skl w_k,w_j\skr_{L^{2}} ,\qquad A_{jk} = a\big(w_k,w_j \big)\qquad \text{for all }j,k=1,\ldots,n
\end{equation}

Let us mention two critical points that arise in our setting.
First, approximation errors in the entries of the stiffness matrix $A$ typically lead to significant numerical errors of the resulting scheme.
 As a consequence, they 
 have to be computed with high precision.
Second, due to the nonlocal nature of  operator $\OA$, the matrix $A$ is fully populated. 
This leads to a high computational cost of the solution scheme, which can be reduced considerably by using compression techniques.
\\[-1ex]




\item[(6)] \emph{Time discretization:} 


Having reached equations \eqref{Pide-matrixform} and \eqref{Pide-matrixform-initial}, we are finally left to solve a linear system of ordinary differential equations. A variety of discretization methods for these types of equations is available, for instance Euler schemes. 
\end{enumerate}

\begin{figure}[htb!]
\centering
\includegraphics[scale=0.8]{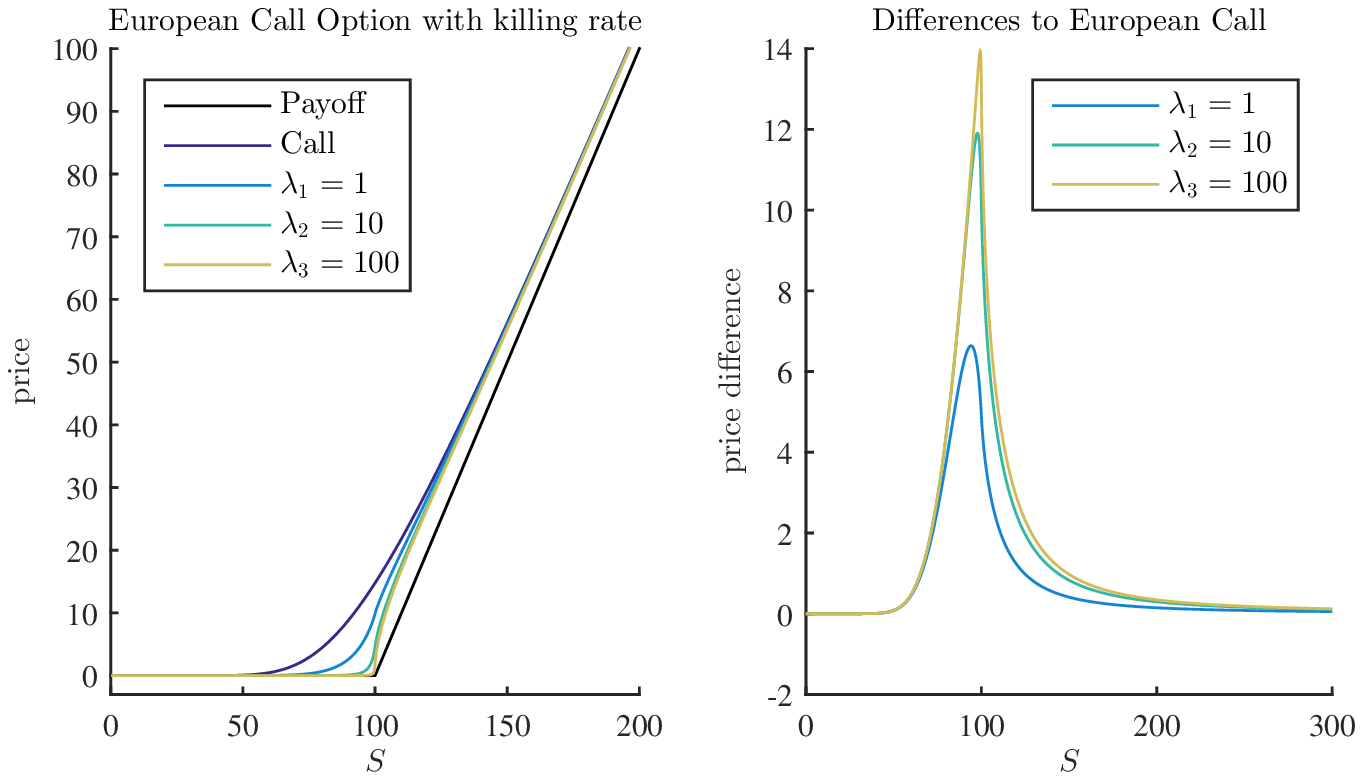}
\includegraphics[scale=0.8]{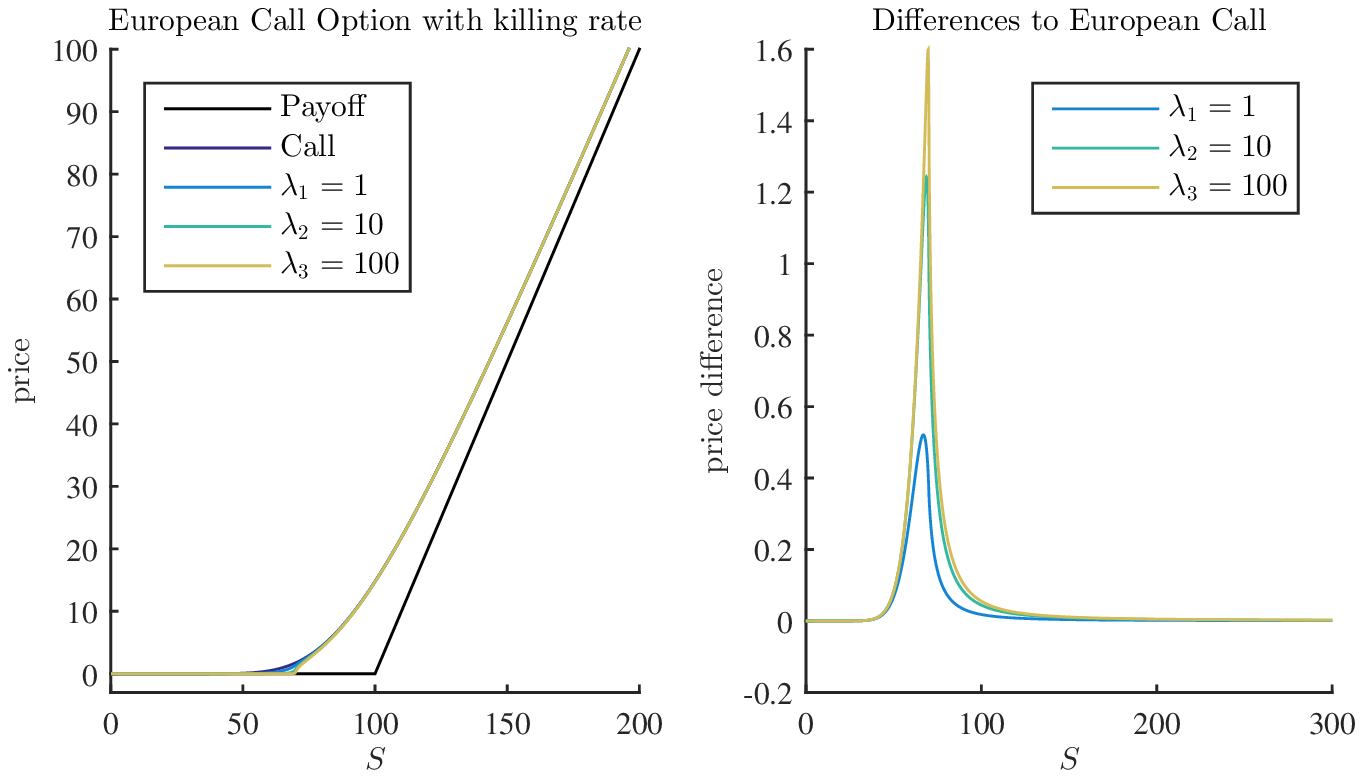}
\caption{Effect of a killing rate of indicator type on a call option price in a pure jump L\'evy model driven by a CGMY process. Top: $B=K=100$. Bottom: $B=70$. Left: payout of the call option along with the prices of the call option and the employee options with $\lambda=1,10,100$. Right: differences between the call price and the prices of the employee\tild options.}\label{fig-kill}
\end{figure}

To illustrate the numerical effect of the killing rate we use an implementation subject to the following
specifications in the steps (1)--(6).\footnote{The author gratefully acknowledges Christoph Schwab and his working group for letting her use their code, which implements the Galerkin method for pricing European call options in CGMY-models.}
\begin{itemize}
\item
The equation is modified according to the choice $\psi(t,x)\coloneqq \max(\ee{x}-K,0)$. Dominated convergence yields $|u(t,x) - \psi(x)|\to 0$ for $|x|\to\infty$. 
For a similar situation 
Proposition 4.1 in \cite{ContVoltch.2005b} 
shows exponential convergence. 
We conjecture that also in our setting we have an exponential decay of the difference $|\phi-\tilde{\phi}|$. \\[-1ex]
\item
As Riesz basis a wavelet basis of first polynomial order 
 is chosen and combined with a compression technique replacing the stiffness matrix by a sparse one. 
We refer to Section~12.2.2 in the monograph of Hilber, Reichmann, Schwab and Winter (2013) \nocite{HilberReichmannSchwabWinter2013} for a presentation of the wavelet compression technique and to \cite{PetersdorffSchwab2003} for a related error\tild analysis. \\[-1ex]
\item
As time discretization an $hp$-discontinuous Galerkin method is chosen as the initial condition is not differentiable 
and a scheme selecting more time points at the beginning is advantageous. 
For details we refer to
 Section\tild 12.3 in \cite{HilberReichmannSchwabWinter2013}.\\[-0.9ex]
\end{itemize}


In our numerical experiments we consider different domains on which the killing rate is active. Each of these domains is specified by a parameter~$B$ according to $\tilde{\kappa}(x)=\lambda\1_{(-\infty, \log(B))}(x)$. Figure~\ref{fig-kill} depicts our results for~$B=K=100$ and~$B=70$. In both cases the maturity (in years) is set to $T=1$ and the strike to $K=100$. The parameters of the process are set to $C = 0.01560$, $G = 0.0767$, $M = 7.55$ and $Y = 1.2996$.
Notice that the scales of the graphs on the left side are chosen differently from those on the right side. 

We see that the killing rate of indicator type displays an effect in all of the considered cases. While the difference between the call and the employee options peaks around the level $S_0=B$, the killing rate affects prices globally with fast decay 
on both sides. The effects are stronger for higher scale parameters\tild $\lambda$. This leads to a monotone order of the price curves, the higher the scale parameter $\lambda$, the lower the price when all other parameters are kept equal.

\section{Robustness of the weak solutions}\label{app-robust}

We provide a robustness result that shows that small perturbations of the data $f$ and $g$ and, more critically, of the bilinear form $a$ only have a small effect on the weak solution of Kolmogorov equation \eqref{parabolic-eq-origin}. The result is crucial for the limit procedure in our derivation of the Feynman-Kac-type representation in Theorem \ref{fkac}.

Let $X\hookrightarrow H\hookrightarrow X^\ast$ be a Gelfand triplet.
For $t\in[0,T]$ and each $n\in\nn$ let\tild $\OA_t^n$ respectively $\OA_t$ be an operator with associated real-valued bilinear form $a^n_t$ respectively $a_t$. We introduce the following set of conditions.
{
\begin{enumerate}[label=(An\arabic{*}),leftmargin=3em]
\item\label{An1}  There exists a constant $C_1>0$ such that uniformly for all $n\in\nn$, $t\in[0,T]$ and $u,v\in X$,
\begin{align}\label{contX}
\max\big\{\big|a_t^n(u,v)\big|,\big|a_t(u,v)\big|\big\}&\le C_1\|u\|_{X}\|v\|_{X}.
\end{align}
\item There exists constants $C_2, C_3>0$ such that uniformly for all $n\in\nn$, $t\in[0,T]$ and $u\in X$,
\begin{align}\label{gardn}
\min\{a^n_t(u,u), a_t(u,u)\}&\ge C_2\|u\|^2_{X}-C_3\|u\|^2_{H}.
\end{align}
\item\label{An3} There exists a sequence of functionals $F_n:L^2(0,T;H)\to\rr_+$ 
such that for all $n\in\nn$ and $u,v\in L^2(0,T;H)$, both $F_n(u)\to 0$ for $n\to\infty$ and
\begin{align}\label{an-a-conv}
\int_0^T\!\big|(a^n_t-a_t)(u(t),v(t))\big|\dd t\le F_n(u)\|v\|_{L^2(0,T;H)}.
\end{align}
\end{enumerate}
}
\begin{lemma}\label{lem-robust}
Let operators $\OA$ and $\OA^n$ for $n\in\nn$ satisfy \ref{An1}--\ref{An3}.
Let $f^n,f\in L^2(0,T;H)$ with $f^n\to f$ in $L^2\big(0,T;X^\ast)$ and $g^n,g\in H$ with $g^n\to g$ in $H$. Then the sequence of unique weak solutions $u^n\in W^1(0,T;Y,H)$ of
\begin{equation}\label{eq-un}
\dot{u}^n + \OA^n_t u^n = f^n,\quad u^n(0)=g^n
\end{equation}
converges strongly in $L^2\big(0,T;X)\cap C(0,T;H)$ to the unique weak solution $u \in W^1(0,T;X,H)$ of
\begin{equation}\label{eq-gwu}
\dot{u} + \OA_t u = f,\quad u(0)=g.
\end{equation}
\end{lemma}

\begin{proof}
Fix some $n\in \nn$ and let $u^n,u\in W^1(0,T;X,H)$ be the unique weak solutions of equations \eqref{eq-un} and \eqref{eq-gwu} and let $w^n:=u-u^n$. Substracting equation\tild \eqref{eq-un} from \eqref{eq-gwu} and inserting $w^n$ as test function yields for every $t\in[0,T]$,
\begin{align}\label{test-withwn}
\begin{split}
\lefteqn{
\int_0^t \big(\dot{w}^n(s),w^n(s)\big) \dd s + \int_0^t a^n_s\big(w^n(s),w^n(s) \big) \dd s}\qquad\\
&=
\int_0^t\big(f^n(s)-f(s),w^n(s)\big) \dd s + \int_0^t \big(a^n_s-a_s)(u(s),w^n(s)\big)\dd s.
\end{split}
\end{align}
We insert $\int_0^t \big( \dot{w}^n(s), w^n(s) \big) \dd s=\frac{1}{2} \big( \|w^n(t)\|_{H}^2 - \|w^n(0)\|_{H}^2 \big)$,
see e.g.\ \cite{Wloka-english} (equation (2) on p. 394),  inequalities \eqref{gardn}, \eqref{an-a-conv} and the inequality of Young. Subsequently applying the lemma of Gronwall yields the existence of constants $c_1,c_2>0$ such that
 \begin{align}
\begin{split}
\lefteqn{
\sup\limits_{t\in[0,T]} \|w^n(t)\|^2_H + c_1\|w^n\|^2_{L^2(0,T;X)}
}\qquad\quad\\
&\le c_2 \Big( \big|F_n(u)\big|^2 + \|f^n-f\|^2_{L^2(0,T;X^\ast)} + \|g^n-g\|^2_H \Big)
\end{split}
 \end{align}
 with $F_n$ from condition (An3). Hence $u^n\to u$ converges strongly in $L^2\big(0,T;X)$ and in $C(0,T;H)$, which proves the lemma.
 \qed
\end{proof}

\section{Proof of the Feynman-Kac-type formula, part (ii) of Theorem \ref{fkac}}\label{sec-proof-fkacgeneral}
The key steps in the proof of the Feynman-Kac-type formula in Theorem \ref{fkac} are first applying It\^o's formula with the help of the regularity assertion in Lemma \ref{lem-regCinfty} below, second invoking the convergence of regularized solutions to the solution of Kolmogorov equation \eqref{gl-stochdarst-origin} due to robustness result Lemma \ref{lem-robust}, and third linking convergence in $L^2_\eta(\rrd)$ respectively $L^2\big(0,T;H^l_\eta(\rrd)\big)$ to the convergence of conditional expectations via Lemma\tild \ref{Ephi_kl_cphi}.

\begin{lemma}\label{lem-regCinfty}
For  $\eta\in\rrd$ and $\alpha>0$,
let $\OA$ be a pseudo differential operator whose symbol $A$ 
has Sobolev index $\alpha$ uniformly in $[0,T]\times R_{\eta}$ 
and let the mapping $t\mapsto A_t(\xi-i\eta)$ be continuous for every $\xi\in\rrd$. 
For $\kappa\in L^\infty([0,T]\times\rrd)$, $g\in L^2_\eta(\rrd)$ and $f\in L^2\big(0,T ;H^{\alpha/2}_\eta(\rr^d)\big)$, let $u\in W^1\big(0,T ;H^{\alpha/2}_\eta(\rr^d), L^2_\eta(\rrd)\big)$ be the unique weak solution of
\begin{align}\label{para-eq-lem-1}
\dot{ u} + \OA_{T-t}u + \kappa_{T-t} u &= f, \\
u(0) &=\, g.\label{para-eq-lem-initial}
\end{align}
Then the following assertions hold.
\begin{enumerate}[label=$(\roman*)$,leftmargin=2em]
\item 
Let $m\ge1$. If $g\in H^{(m-1)\alpha/2}_\eta(\rrd)$, $f\in L^2\big(0,T ;H^{(m-1)\alpha/2}_\eta(\rr^d)\big)$ and $\kappa h \in L^2\big(0,T;H^{k\alpha/2}_\eta(\rrd)\big)$ for all 
$1\le k\le m$ and $h \in L^2\big(0,T;H^{k\alpha/2}_\eta(\rrd)\big)$, 
then $u\in L^2\big(0,T;H^{m\alpha/2}_\eta(\rrd)\big)$ and $\dot{u}\in L^2\big(0,T;H^{(m-2)\alpha/2}_\eta(\rrd)\big)$.
\vspace{1ex}
\item If $g\in H^{\beta}_\eta(\rrd)$ for $\beta=m+d/2+\max(\alpha,1/2)$, $f\in L^2\big(0,T;H^{\gamma}_\eta(\rrd)\big)$ for $\gamma=m+(d+1)/2$ and $\kappa\in C^{\infty}_0([0,T]\times\rrd)$, then for every multiindex $k=(k_1,\ldots,k_d)$ with $|k|\le m$ the derivative $(1+\partial_t)D^k u$ is in $C([0,T]\times\rrd)$. If, moreover, $\OA$ is the Kolmogorov operator of a L\'evy process and $f$ is continuous, then equation \eqref{para-eq-lem-1} holds pointwise for all $(t,x)\in(0,T]\times\rrd$.
\end{enumerate}
\end{lemma}
\begin{proof}
We derive the regularity assertion by explicit operations on the Fourier transform of the unique weak solution $u\in W^1\big(0,T; H^{\alpha/2}_\eta(\rr^d),L^2_\eta(\rr^d)\big)$ of equations \eqref{para-eq-lem-1} and \eqref{para-eq-lem-initial}. 
 We show the identity
\begin{align}\label{u=sumui}
u = \tilde{u} := u^1 + u^2 + u^3
\end{align}
with
\begin{align*}
\OF_\eta\big(u^1(t)\big) &\coloneqq \OF_\eta(g)\ee{-\int_{T-t}^{T} A_u(\cdot-i\eta) \dd u},
\\
\OF_\eta\big(u^2(t)\big) &\coloneqq \int_0^t \OF_\eta\big(f(s)\big)\ee{-\int_{T-t}^{T-s} A_u(\cdot-i\eta) \dd u} \dd s,
\\
\OF_\eta\big(u^3(t)\big) &\coloneqq - \int_0^t \OF_\eta\big(\kappa u (s)\big)\ee{-\int_{T-t}^{T-s} A_\lambda(\cdot-i\eta) \dd \lambda} \dd s
\end{align*}
and hence
\begin{align}
\partial_t\OF_\eta\big(u^1(t)\big) &= - A_{T-t}(\cdot-i\eta) \OF_\eta\big(u^1(t)\big),\nonumber
\\
\partial_t \OF_\eta\big(u^2(t)\big) &= - A_{T-t}(\cdot-i\eta) \OF_\eta\big(u^2
(t)\big) + \OF_\eta\big(f(t)\big),\nonumber
\\
\partial_t \OF_\eta\big(u^3(t)\big) &=  - A_{T-t}(\cdot-i\eta) \OF_\eta\big(u^3(t)\big) - \OF_\eta\big(\kappa u(t)\big). \nonumber
\end{align}
In particular, $\tilde{u}$ satisfies equation \eqref{para-eq-lem-1}.
Inequality \eqref{absch-Phuteta} from Remark \ref{rem-implicationsA1A3} with constants $C_1,C_2>0$ and the
inequality of Cauchy-Schwarz guarantee the existence of constants $c_1,c_2>0$ which are such that for all $(t,\xi)\in[0,T]\times\rrd$, 
\begin{align*}
 \big|\OF_\eta\big(u^1(t)\big)(\xi) \big| 
 &\le  C_1 \big|\OF_\eta(g)(\xi) \big|\ee{-t C_2|\xi|^\alpha},
 \\
 \big|\OF_\eta\big(u^j(t)\big)(\xi) \big| &\le  C_1 \left(\int_0^t \big|\OF_\eta(f(s))(\xi)\big|^2 \dd s \right)^{1/2} \left(\int_0^t \ee{-(t-s)2C_2|\xi|^\alpha}\dd s \right)^{1/2}
 \\
 &\le 
c_2 \left(\int_0^T \big|\OF_\eta(f^j(s))(\xi) \big|^2\big( 1 + |\xi|\big)^{-\alpha}\dd s\right)^{1/2}
\end{align*}
as well as
\begin{align*}
 \big|\OF_\eta\big(\partial_t u^1(t)\big)(\xi) \big| 
 &\le  c_1 \big|\OF_\eta(g)(\xi) \big|\big( 1 + |\xi|\big)^{\alpha}\ee{-t C_2|\xi|^\alpha},
 \\
 \big|\OF_\eta\big(\partial_t u^j(t)\big)(\xi) \big| 
 &\le 
c_2 \left(\int_0^T \big|\OF_\eta(f^j(s))(\xi) \big|^2\big( 1 + |\xi|\big)^{\alpha}\dd s\right)^{1/2} + \big|\OF_\eta(f^j(s))(\xi)\big|,
\end{align*}
for $j=1,2$ with $f^1=f$ and $f^2=-\kappa u$.
Hence there is a constant $c_4>0$ with
\begin{align*}
\lefteqn{\|\tilde{u}\|_{L^2(0,T;H^{m\alpha/2}_\eta(\rrd))} + \|\partial_t \tilde{u}\|_{L^2(0,T;H^{(m-2)\alpha/2}_\eta(\rrd))}}\quad\\
&\scalebox{.97}[1]{$\displaystyle \le 
c_4\big(\|g\|_{H^{(m-1)\alpha/2}_\eta(\rrd))} \! + \|f\|_{L^2(0,T;H^{(m-1)\alpha/2}_\eta(\rrd))} \! + \|\kappa u\|_{L^2(0,T;H^{(m-1)\alpha/2}_\eta(\rrd))}\big) .$}
\end{align*}
For $m=1$, by inserting $u\in L^2\big(0,T;H^{\alpha/2}_\eta(\rrd)\big)$ and $\kappa u \in L^2\big(0,T;H^{\alpha/2}_\eta(\rrd)\big)$,
we obtain $\tilde{u}\in L^2\big(0,T;H^{\alpha/2}_\eta(\rrd)\big)$ and $\partial_t\tilde{u}\in L^2\big(0,T;H^{-\alpha/2}_\eta(\rrd)\big)$.
In particular, $\tilde{u}\in W^1\big(0,T;H^{\alpha/2}_\eta(\rrd),L^2_\eta(\rrd)\big)$ is the unique weak solution $\tilde{u}=u$ of equations \eqref{para-eq-lem-1} and \eqref{para-eq-lem-initial}.

For $m=2$ it is thus sufficient to notice that $\kappa u \in L^2\big(0,T;H^{\alpha/2}_\eta(\rrd)\big)$ implies $u\in L^2\big(0,T;H^{\alpha/2}_\eta(\rrd)\big)$ and $\partial_t\tilde{u}\in L^2\big(0,T;L^2_\eta(\rrd)\big)$. An iterative argument then yields part $(i)$ of the Lemma.
\par
$(ii)$ $\vphantom{l}$
By the inequality of Cauchy-Schwarz and $\int_{\rrd}\big( 1 + |\xi|\big)^{-d-\epsilon}\dd\xi<\infty$ if $\epsilon>0$, we obtain that for $\beta=m+d/2 + \max(\alpha,1/2)$ and $\gamma=m+(d+1)/2$ there exists a constant $c_5>0$ such that
\begin{align*}
\lefteqn{\int_{\rrd}\big|\big(1 + \partial_t\big) \OF_\eta\big((u^1+u^2)(t)\big)(\xi) \big|\big( 1 + |\xi|\big)^m\dd \xi}\qquad\qquad\qquad\qquad\\
&\le 
c_5\big(\|g\|_{H^{\beta}_\eta(\rrd)} +\|f\|_{L^2(0,T;H^{\gamma}_\eta(\rrd))}  \big) <\infty.
\end{align*}

Furthermore, the mappings $t\mapsto \OF_\eta\big(\tilde{u}(t)\big)(\xi)$ and $t\mapsto \partial_t \OF_\eta\big(\tilde{u}(t)\big)(\xi)$ are continuous for each $\xi\in\rrd$. 
Dominated convergence implies $D^k_x(1+\partial_t)(u^1+u^2)\in C([0,T]\times \rrd)$ for every multiindex $k=(k_1,\ldots,k_d)$ with $|k|\ge0$.
Moreover, there exists a constant $c_6>0$ such that
\begin{align*}
\int_{\rrd}\big|\big(1 + \partial_t\big) \OF_\eta\big(u^3(t)\big)(\xi) \big|\big( 1 + |\xi|\big)^m\dd \xi
&\le 
c_6\|\kappa u\|_{L^2(0,T;H^{\gamma}_\eta(\rrd))} <\infty.
\end{align*}
Dominated convergence yields $D^k_x(1+\partial_t)u^3\in C([0,T]\times \rrd)$ for every multiindex $k=(k_1,\ldots,k_d)$ with $|k|\ge0$.


Now let $\OA$ be the Kolmogorov operator of a time-inhomogeneous L\'evy process.
In order to establish equation \eqref{para-eq-lem-1} pointwise, fix a $t\in T$ for which the equation holds (as operator equation) and choose a sequence $u_n\in C^\infty_0((0,T)\times\rrd)$ such that $u_n(t)\to u(t)$ in the norm of $H^{\alpha/2}_\eta(\rrd)$ and $\dot{u}_n(t)\to \dot{u}$ in the norm of $L^2_\eta(\rrd)$. Moreover, let $\varphi\in C^\infty_0(\rrd)$. We notice that $\OA_{T-t} u(t)$ is defined pointwise, since $u(t)\in C^2(\rrd)$.  An elementary manipulation and the continuity of the scalar product yield
\begin{align*}
\int_{\rrd} \OA_{T-t} u(t,x) \varphi(x) \ee{- 2 \skl \eta,x\skr}\dd x
&=
\skl u(t), \OA^{-\eta,\ast}_{T-t}\varphi\skr_{L^2_\eta} 
=
 \lim_{n\to \infty}\skl u_n(t), \OA^{-\eta,\ast}_{T-t}\varphi\skr_{L^2_\eta}
\end{align*}
with the adjoint operator $\OA^{-\eta,\ast}_{T-t}$ defined in Lemma \ref{lem-adjointA} in Appendix \ref{sec-adjop}.
Equation\tild \eqref{def-A} from the same lemma 
and the continuity of the bilinear form imply
\begin{align*}
 \lim_{n\to \infty}\skl u_n(t), \OA^{-\eta,\ast}_{T-t}\varphi\skr_{L^2_\eta}
 =
 \lim_{n\to \infty} a_{T-t}(u_n(t),\varphi)
 = a_{T-t}(u(t),\varphi)
\end{align*}
and hence
\begin{align*}
\skl \dot{u}(t), \varphi\skr_{L^2_\eta} + \skl \OA_{T-t} u(t), \varphi\skr_{L^2_\eta} &= \skl f(t), \varphi\skr_{L^2_\eta}\qquad \text{for all }\varphi\in C^\infty_0(\rrd).
\end{align*}
From the fundamental lemma of variational calculus, $\dot{u}(t,x) + \OA_{T-t} u(t,x)=f(t,x)$ follows for a.e. $x\in \rrd$. Since we can choose $t$ arbitrarily from a dense subset in $(0,T)$, the assertion follows by continuity of $\dot{u} + \OA_{T-\cdot} u - f$.
\qed
\end{proof}

\begin{lemma}\label{Ephi_kl_cphi}
For $\eta\in\rrd$ and $\alpha\in(0,2]$,
let $L$ be a time-inhomogeneous L\'evy process with symbol $A=(A_t)_{t\in[0,T]}$ that satisfies exponential moment condition\tild (A1) and  G{\aa}rding condition (A3). Then
\begin{enumerate}[label=$(\roman{*})$, leftmargin=2em]
 \item for every $t>0$ there exists a constant $C(t)>0$ such that $E[|\varphi(L_s)|]\le C(t) \|\varphi\|_ {L^2_\eta(\rrd)}$ uniformly for all $\varphi\in L^2_\eta(\rr^d)$ and  $s\in[t,T]$,\vspace{1ex}
 \item
for $l>(d-\alpha)/2$ and every $0\le t<T$ there exists a constant $C_1>0$ such that $ \big| E \big( \int_t^T \varphi(s,L_s) \dd s \big| \,\OF_t \big) \big|\le C_1  \| \varphi \|_{L^2(t,T;H^l_\eta(\rr^d))}$
uniformly for all $\varphi \in L^2\big(0,T; H^{l}_\eta(\rr^d)\big)$.
\end{enumerate}
\end{lemma}
\begin{proof}
$(i)$ $\vphantom{l}$
By Remark \ref{rem-implicationsA1A3} and Condition (A3), the distribution of $L_t$ has a Lebesgue density. Applying Parseval's identity, 
we obtain
\begin{align*}
E|\varphi(L_t)|
&=
\frac{1}{(2\pi)^d} \int_{\rrd} \OF(|\varphi|)(\xi - i\eta) \overline{\ee{-\int_0^tA_s(\xi-i\eta)\dd s}} \dd \xi.
\end{align*}
Inserting inequality \eqref{absch-Phuteta} and the inequality of Cauchy-Schwarz then yields assertion~$(i)$.
\par
$(ii)$ $\vphantom{l}$
W.l.o.g. $\varphi\ge0$. We have
$ E \big( \int_t^T \varphi(s,L_s) \dd s \big| \,\OF_t \big)=G(L_t)$
with
\begin{align*}
G(y)
&=
E\Big( \int_0^{T-t} \varphi(s+t, L_{t+s}-L_t +y) \dd s \Big).
\end{align*}
The theorem of Fubini and Parseval's identity imply
\begin{align*}
G(y)
&=
\frac{1}{(2 \pi)^d}
\int_{\rrd} \int_0^{T-t} \OF\big(\tau_{y}\varphi(s+t)\big)(\xi-i\eta) \overline{\ee{-\int_0^{s}A_{t+u}(\xi-i\eta)\dd u}}  \dd s \dd \xi,
\end{align*}
where $\tau_yf(x)\coloneqq f(x+y)$. Notice that $\OF_\eta(\tau_y f)(\xi)=\ee{-\skl \xi,y\skr}\OF_\eta(f)(\xi)$. Inserting the inequality of Cauchy-Schwarz and equation \eqref{absch-Phuteta} with constants $C_1,C_2>0$, we obtain that for $l>d-\alpha$ there are constants $c_1,c_2>0$ such that
\begin{align*}
|G(y)|& \le C_1 \int_{\rrd} \left( \int_0^{T-t} \big|\OF\big(\tau_y\varphi(s+t)\big)(\xi-i\eta)\big|^2\dd s \int_0^{T-t}\! \ee{- 2sC_2|\xi|^\alpha} \dd s\right)^{\!1/2}\!\!\dd \xi
\\
&\le c_1\int_0^T \int_{\rrd}\big|\OF\big(\varphi(s+t)\big)(\xi-i\eta)\big|\big(1+|\xi|\big)^\alpha \dd \xi\dd s
\\
&\le c_2\|\varphi\|_{L^2(0,T;H^{l/2}_\eta(\rrd))}.
\end{align*}
This concludes the proof.
\qed
\end{proof}

We are now in a position to prove 
part $(ii)$ of Theorem \ref{fkac}.
\begin{proof}[Theorem \ref{fkac}, part $(ii)$]
First, assume that $t\mapsto A_t(\xi-i\eta)$ is continuous for all $\xi\in\rrd$.
By density arguments, respectively mollification, we can choose sequences $g^n\in C^\infty_0(\rrd)$, $f^n\in C^\infty_0\big([0,T]\times\rrd\big)$ and $\kappa ^n \in  C^\infty_0\big([0,T]\times\rrd\big)$ such that for $n\to \infty$,
\begin{align*}
g^n &\rightarrow g\quad\text{in }L^2_\eta(\rrd),\\
f^n &\rightarrow f\quad\text{in }L^2\big(0,T;H^l_\eta(\rrd)\big),\\
\kappa^n &\rightarrow \kappa\quad\text{pointwise  and }\sup\limits_{n\in \nn}\|\kappa^n\|_{L^\infty([0,T]\times\rrd)} < \infty.
\end{align*}
We denote, with a slight abuse of notation, by $a^n_t$ the bilinear form associated with $\OA_t + \kappa^n_t$ and $a_t$ the bilinear form associated with $\OA_t + \kappa_t$. Then
\begin{align}\label{an=a+}
a^n_t(u,v) = a_t(u,v) + \skl (\kappa^n_t-\kappa_t) u,v\skr_{L^2_\eta(\rrd)}\quad\text{for all }u,v\in H^{\alpha/2}_\eta(\rrd).
\end{align}
Together with Conditions (A1)--(A3) and the uniform boundedness of $\kappa^n$ and $\kappa$, we obtain the validity of Conditions (An1) and (An2) from Section 6. Moreover, by equality \eqref{an=a+} and the Cauchy-Schwarz inequality, we get
\begin{align}
\Bigg|\int_0^T(a^n-a)\big(u(s),v(s)\big) \dd s \Bigg| \le \big\|(\kappa^n-\kappa)u\big\|_{L^2(0,T;L^2_\eta(\rrd))} \| v\|_{L^2(0,T;L^2_\eta(\rrd))}.
\end{align}
It follows from pointwise convergence of $\kappa^n\to\kappa$ and dominated convergence that for $n\to \infty$,
\[
F_n(u):=\big\|(\kappa^n-\kappa)u\big\|_{L^2(0,T;L^2_\eta(\rrd))}\to 0 \quad \text{for all }u\in L^2(0,T;L^2_\eta(\rrd)),
\]
and hence Condition (An3) is satisfied.
Let $u^n\in W^1\big(0,T;H^{\alpha/2}_\eta(\rrd), L^2_\eta(\rrd)\big)$ be the unique weak solution of
\begin{align}\label{eq-un-proof}
\dot{ u}^n + \OA_{t}u^n + \kappa_{t}^n u^n &= f^n, \quad
u^n(0) =\, g^n.
\end{align}
Lemma \ref{lem-robust} yields the convergence $u^n\rightarrow u$, both in the space $L^2\big(0,T;H^{\alpha/2}_\eta(\rrd)\big)$ and in $ C\big(0,T;L^2_\eta(\rrd)\big)$, to the weak solution $u\in W^1\big(0,T;H^{\alpha/2}_\eta(\rrd), L^2_\eta(\rrd)\big)$ of
%
%
\begin{align}\label{eq-u-proof}
\dot{ u} + \OA_{t}u + \kappa_{t} u &= f, \quad
u(0) =\, g.
\end{align}
Lemma \ref{lem-regCinfty} shows that equality \eqref{eq-un-proof} holds pointwise and that $u^n$ is regular enough to apply It\^o's formula. We denote by $\big(b_t,\sigma_t, F_t;h\big)_{t\in[0,T]}$ the characteristics of $L$ and set $w^n(t,x):=u^n(T-t,x)$. Then
It\^o's formula for semimartingales, see for instance Theorem I.4.57 in \cite{JacodShiryeav2003}, entails 
\begin{eqnarray}\label{kern-argument-Ito}
\begin{split}
\lefteqn{ w^n(T,L_T) \ee{-\int_0^T \kappa^n_\lambda(L_{\lambda})\dd \lambda}- w^n(s,L_s)\ee{-\int_0^s \kappa^n_\lambda(L_{\lambda})\dd \lambda} }\,\, \\
&=
\int_s^T \Big[ \dot{w^n} -\OA_\tau w^n - \kappa w^n\Big] (\tau,L_{\tau}) \ee{-\int_0^\tau \kappa^n_\lambda(L_{\lambda})\dd \lambda} \dd \tau \\
& \quad+ \int_s^T \big( \sigma_\tau^{1/2}\cdot \nabla w^n(\tau,L_{\tau}) \big)\ee{-\int_0^\tau \kappa^n_\lambda(L_{\lambda})\dd \lambda}\dd W_\tau \\
&\quad + \Big( \ee{-\int_0^\cdot \kappa^n_\lambda(L_{\lambda})\dd \lambda} \big( w^n(\cdot,L_{\cdot- }+x) - w^n(\cdot,L_{\cdot- }) \big) \1_{(s,\infty)}(\cdot)\Big) \ast\big( \mu - \nu \big)_T.
\end{split}
\end{eqnarray}
Thanks to our assumptions on $g^n,f^n$ and $\kappa^n$, we may decompose $u^n$ in three summands as in equation \eqref{u=sumui}. Then, applying of part $(ii)$ of Lemma \ref{lem-regCinfty}, it is elementary to conclude that $w^n$ and $\nabla w^n$ belong to $L^2(\rrd)$.
Hence, the integrals with respect to $W$ and $\mu-\nu$ are martingales, compare Theorem II.1.33 a) in \cite{JacodShiryeav2003}. 

We insert the identity $\dot{w}^n - \OA_\tau w^n - \kappa^n w^n= \overline{f}^n$ with $\overline{f}(t,x)\coloneqq f^n(T-t,x)$ in equation \eqref{kern-argument-Ito}, subsequently multiply the equation with the term $\ee{\int_0^s \kappa^n_\lambda(L_{\lambda})\dd \lambda}$ and lastly take the conditional expectation. This gives us for $0\le s\le T$,
\begin{align}\label{Ewn}
\begin{split}
 \lefteqn{E\Big(w^n(T,L_T) \ee{-\int_s^T \kappa^n_\lambda(L_{\lambda})\dd \lambda}\Big| \,\OF_s \Big)  - w^n(s,L_s)} \qquad\qquad \\
&=\,
E\Big(\int_s^T \overline{f}^n (\tau,L_{\tau}) \ee{-\int_s^\tau \kappa^n_\lambda(L_{\lambda})\dd \lambda} \dd \tau \Big|\,\OF_s\Big) \,.
\end{split}
\end{align}
Let w.l.o.g. $0< s\le T$. We will now derive the desired stochastic representation by letting $n\to \infty$ for each term in equation \eqref{Ewn}.
Denote $w(t,x):=u(T-t,x)$.
From the convergence $w^n(s,\cdot)\to w(s,\cdot)$ in $L^2_\eta(\rrd)$ and part $(i)$ of Lemma \ref{Ephi_kl_cphi} for $s>0$, we get the convergence 
\[
w^n(s,L_s)\to w(s,L_s)\text{ in $L^1(P)$ and a.s. for a subsequence.}
\]
The pointwise convergence $\kappa^n\to\kappa$ and the uniform boundedness together with dominated convergence imply both $\int_a^b \kappa^n_\lambda(L_{\lambda})\dd \lambda \rightarrow \int_a^b \kappa_\lambda(L_{\lambda})\dd \lambda $ and uniform boundedness of the sequence for $0\le a\le b\le T$.

Together with $w^n(s,L_s)\to w(s,L_s)$ in $L^1(P)$ the convergence
\begin{align*}
 E\Big( \big|w^n(t,L_t)\ee{-\int_s^T \kappa^n_\lambda(L_{\lambda})\dd \lambda} - w(t,L_t)\ee{-\int_s^T \kappa_\lambda(L_{\lambda})\dd \lambda}\big|\Big| \,\OF_s \Big)
  \rightarrow 0
\end{align*}
as $n\to \infty$ then follows from the triangle inequality.

Next, denote $\overline{f}(t,x)\coloneqq f(T-t,x)$. Since $f^n\rightarrow f \in L^2\big(t,T;H^l_\eta(\rr^d)\big)$, part\tild $(ii)$ of Lemma \ref{Ephi_kl_cphi} guarantees the existence of a constant $c_2>0$ for  $l>(d-\alpha)/2$ such that
\begin{align*}
E\Big(\int_s^T \big|(\overline{f}^n -\overline{f})(h,L_{h}) \big| \dd h \Big|\,\OF_s\Big)
\le c_1 \| f^n-f \|_{L^2(t,T;H^l_\eta(\rr^d))} \rightarrow 0.
\end{align*}
Now the triangle inequality yields the convergence of the second line in equation~\eqref{Ewn}
 and therefore part $(ii)$ of Theorem \ref{fkac} under the additional assumption that the mapping $t\mapsto A_t(\xi-i\eta)$ is continuous for every $\xi\in\rrd$.
 Finally, thanks to the tower rule of conditional expectation, 
 the claim follows by induction over the continuity periods also under the more general assumption that $t\mapsto A_t(\xi-i\eta)$ is c\`adl\`ag for every $\xi\in\rrd$.
 \qed
\end{proof}

\section{Acknowledgements}
The roots of the present paper go back to the author's dissertation \cite{PhdGlau}, which was financially supported by the DFG through project \mbox{EB66/11-1}. The author expresses her gratitude to Ernst Eberlein for his valuable support. 
The author also gratefully acknowledges Christoph Schwab and his working group for letting her use their wavelet-Galerkin implementation.
She furthermore thanks Carsten Eilks, Paul Harrenstein, Alexandru Hening, Claudia Kl{\"u}ppelberg for valuable discussions and comments and Wolfgang Runggaldier and the anonymous referees for their suggestions to improve the manuscript.
%

  \appendix\normalsize
\section{Adjoint Operator}\label{sec-adjop}
For a L\'evy process $L$ with characteristics $(b,\sigma,F;h)$ we denote by $\OA^{(b,\sigma,F)}$ and $A^{(b,\sigma,F)}$ its Kolmogorov operator and its symbol, respectively. As the following assertions straightforwardly extend to the time-inhomogeneous case, we here present the time-homogeneous case only.
\begin{lemma}\label{lem-Aeta}
For $\eta\in\rrd$,
let $L$ be a L\'evy process with characteristics $(b,\sigma,F;h)$ that  satisfies exponential moment condition $EM(\eta)$ and denote by $\OA$ its Kolmogorov operator with symbol $A$. Then
\[
A^\eta(\xi ):= A (\xi + i\eta) = A^{(b^\eta,\sigma,F^\eta)} (\xi) + A (i\eta)\qquad\text{for all $\xi\in\rrd$}
\]
with
\begin{align*}
b^\eta &= b + \sigma \cdot \eta + \int_{\rrd} \big( \ee{\skl \eta,y\skr} - 1 \big) h(y) F(\dd y),\\
F^\eta(\dd y) &= \ee{\skl \eta, y\skr} F(\dd y).
\end{align*}
In particular, $A^\eta$ is the symbol of a L\'evy process with killing rate $A (i\eta)$. Moreover, its Kolmogorov operator $\OA^\eta$ satisfies
\[
\OA^\eta \varphi = \ee{-\skl \eta, \cdot\skr} \OA(\ee{\skl \eta, \cdot\skr}\varphi) = \OA^{(b^\eta,\sigma,F^\eta)} \varphi + A (i\eta)\varphi \qquad\text{for all }\varphi\in C^\infty_0(\rrd).
\]
%
\end{lemma}
\begin{proof}
It is elementary to verify the assertion for the symbol. This can then nicely be used to verify the assertion for the operator: Let 
 $\varphi\in C^\infty_0(\rrd)$, then $\OF(\ee{\skl \eta,\cdot\skr}\varphi)(\xi)=  \OF(\varphi)(\xi-i\eta)$ and 
 \begin{align*}
 \OA\big( \ee{\skl \eta,\cdot\skr}\varphi\big) (x) &= \OF^{-1}\big( A \OF(\ee{\skl \eta,\cdot\skr}\varphi)\big)\\
 &=\frac{1}{(2\pi)^d} \int_{\rrd} \ee{-i\skl \xi,x\skr} A(\xi) \OF(\varphi)(\xi - i\eta) \dd \xi\\
 &= \frac{ \ee{\skl \eta,x\skr}}{(2\pi)^d} \int_{\rrd} \ee{-i\skl \xi,x\skr} A(\xi+i\eta) \OF(\varphi)(\xi) \dd \xi,
 \end{align*}
 which concludes the proof.
\qed
\end{proof}
For all $\varphi\in C^\infty_0(\rrd)$ let
\[
\OF_\eta(\varphi) \coloneqq \ee{-\skl \eta,\cdot \skr} \OF\big(\varphi\ee{\skl \eta,\cdot \skr}\big)\,\text{
and }\OF_\eta^{-1}(\varphi) \coloneqq \ee{-\skl \eta,\cdot \skr} \OF^{-1}\big(\varphi\ee{\skl \eta,\cdot \skr}\big).
\]
Theorem 4.1 in \cite{EberleinGlau2013} shows that for a pseudo differential  operator $\OA$ whose symbol $A$ has a continuous extension to $U_{-\eta}$ that is analytic in the interior of $U_{-\eta}$ and satisfies the continuity condition (A2), we have
\begin{equation}
\OA \varphi = \OF^{-1}(\OA\varphi) = \OF^{-1}_\eta(A^{-\eta}\OF_\eta(\varphi))\quad\text{for all }\varphi\in C^\infty_0(\rrd).
\end{equation}
Parseval's equality yields for all $\varphi,\,\psi\in C^\infty_0(\rrd)$,
\begin{equation}
a(\varphi,\psi) = \skl \OA \varphi, \psi\skr_{L^2_\eta} = \frac{1}{(2\pi)^d}\skl A^{-\eta} \OF_\eta (\varphi),  \OF_\eta (\psi)\skr_{L^2_\eta}.
\end{equation}
For the pseudo differential operator $\OA$ and its symbol $A$, we define their $L^2_\eta$-adjoints  $\OA^{\eta,\ast}$ and $A^{-\eta,\ast}$ such that for all $\varphi,\,\psi\in C^\infty_0(\rrd)$,
\begin{align}
\skl \OA, \psi\skr_{L^2_\eta} &= \skl  \varphi,  \OA^{\eta,\ast} \psi\skr_{L^2_\eta},
\\
\skl A^{-\eta} \OF_{\eta}(\varphi), \OF_{\eta}(\psi)\skr_{L^2_\eta} &= \skl  \OF_{\eta}(\varphi), A^{-\eta,\ast}\OF_{\eta}(\psi)\skr_{L^2_\eta}.
\end{align}

\begin{lemma}\label{lem-adjointA}
For $\eta\in\rrd$, let $L$ be a L\'evy process with characteristics $(b,\sigma,F;h)$ that satisfies exponential moment condition $EM(\eta)$ and denote by $\OA$ its Kolmogorov operator with symbol $A$. Set $F^{-\eta,\ast}(B) = F^{-\eta}_{sym}(B) - F^{-\eta}_{asym}(B)$ for Borel sets $B\neq\{0\}$, where $F^{-\eta}_{sym}(B) = \frac{1}{2} F^{-\eta}(B) + F^{-\eta}(-B)$ and $F^{-\eta}_{asym}(B) = F^{-\eta}(B) -F^{-\eta}_{sym}(B)$ and let
$b^{-\eta,\ast} = -b^{-\eta}$.
Then
\begin{align*}
A^{-\eta,\ast} &= \overline{A^{-\eta}} = A^{(b^{-\eta,\ast},\sigma,F^{-\eta,\ast})} + A(-i\eta),\\
\OA^{-\eta,\ast} \varphi &= \ee{-\skl \eta,\cdot\skr} \OA\big(\ee{\skl \eta,\cdot\skr}\varphi \big)=\OA^{(b^{-\eta,\ast},\sigma,F^{-\eta,\ast})} \varphi + A(-i\eta)\varphi.
\end{align*}
Moreover, $F^{-\eta,\ast}$ is a L\'evy measure.
\end{lemma}
\begin{proof}
For every $\varphi \in C^\infty_0(\rrd)$ we have
\begin{align*}
\skl A^{-\eta} \OF_{\eta}(\varphi), \OF_{\eta}(\psi)\skr_{L^2_\eta} 
&= 
\skl A^{-\eta} \OF\big(\ee{\skl \eta,\cdot\skr}\varphi\big), \OF\big(\ee{\skl \eta,\cdot\skr}\varphi\big)\skr_{L^2}\\
&=
\skl \OF\big(\ee{\skl \eta,\cdot\skr}\varphi\big), \overline{A^{-\eta}}\OF\big(\ee{\skl \eta, \cdot\skr}\varphi\big)\skr_{L^2}.
\end{align*}
Since $A(z)\in\rr$ for $z\in \ccd$ with $\Re(z)=0$, by Lemma \ref{lem-Aeta} we obtain
\begin{align*}
 \overline{A^{-\eta}} = \overline{A^{(b^{-\eta}, \sigma, F^{-\eta})}} + A(-i\eta).
\end{align*}
As $A^{(b^{-\eta}, \sigma, F^{-\eta})}$ is the symbol of a L\'evy process, we have
\[
\overline{A^{(b^{-\eta}, \sigma, F^{-\eta})}(\xi)} = A^{(b^{-\eta}, \sigma, F^{-\eta})}(-\xi)\quad\text{for all }\xi\in\rrd,
\]
whence the assertion of the lemma follows directly.
\qed
\end{proof}

\section{Proof of Theorem \ref{Theo-parabolicity}}\label{sec-proof-Theo-parabolicity}
\begin{theorem}\label{Prop-parabolicity}
For $\eta\in\rrd$ and $\alpha\in(0,2]$,
let the symbol $A=(A_t)_{t\in[0,T]}$ of pseudo differential operator $\OA=(\OA_t)_{t\in[0,T]}$ 
be such that for each $t\in[0,T]$,
$A_t$ has a continuous extension on $U_{-\eta}$ that is analytic in the interior $\Ukri_{-\eta}$ and there exist constants $C(t),m(t)>0$ with
\begin{align}\label{somecont-A}
 \big|A_t(z) \big|&\le  C(t) \big( 1 + |z| \big)^{m(t)} \qquad \text{for all }z\in U_{-\eta}.
 \end{align}
Then the following assertions are equivalent.
\begin{enumerate}[label=$(\roman*)$,leftmargin=2.5em]
\item
The operator $\OA$ is parabolic with respect to 
$\big(H^\alpha_{\eta'}(\rrd), L^2_{\eta'}(\rrd)\big)_{\eta'\in R_{\eta}}$ uniformly  in $[0,T]\times R_{\eta}$.
\item
The symbol $A$ has Sobolev index $2\alpha$ uniformly in 
$[0,T]\times R_{\eta}$.
\end{enumerate}
\end{theorem}
The proof of the theorem is a straightforward generalization of the proof Theorem 3.1 in \cite{Glau2013}, where the assertion is proved for symbols that are constant in time and for spaces without weights, that is if $\eta=0$. In order to provide a self-contained presentation we give a detailed proof. 
\begin{proof}[of Theorem \ref{Prop-parabolicity}]
By the assumption on the analyticity of $A$ and inequality\tild \eqref{somecont-A}, we obtain from Theorem 4.1 in \cite{EberleinGlau2013} that for all $t\in[0,T]$, $\eta'\in R_{\eta}$ and $\varphi,\psi\in C^\infty_0(\rrd)$,
\begin{align}\nonumber
a_t(\varphi,\psi) &= \frac{1}{(2\pi)^d}\skl A_t \OF(\varphi), \OF(\psi)\skr_{L^2} 
= \frac{1}{(2\pi)^d}\skl A_t(\cdot - i\eta') \OF_{\eta'}(\varphi), \OF_{\eta'}(\psi)\skr_{L^2_{\eta'}} \\
&=\frac{1}{(2\pi)^d}\int_{\rrd} \overline{A_t(\xi - i\eta') \OF(\varphi)(\xi - i\eta')} \OF(\psi)(\xi - i\eta') \dd \xi.\label{aviaAeta}
\end{align}
This equality yields that (Cont-$A$) implies (Cont-$a$). Together with the following elementary inequalities, it also yields that (G{\aa}rd-$A$) implies\tild (G{\aa}rd-$a$):
For $C_1>0$, $C_2\ge 0$, $0\le \beta <\alpha$ and $0<C_3<C_1$ there exits a constant  $C_4>0$ such that
$
C_1 x^\alpha - C_2 x^\beta \ge C_3 x^\alpha - C_4$ for all $x\ge 0$\tild
and
\begin{align}\label{elementaryineq}
 C_2 |\xi|^{2\alpha} - C_3 (1+|\xi|^2)^{\beta } 
&\ge C_2 |\xi|^{2\alpha} - C_3 '(1+|\xi|^{2\beta }) 
\ge 
c_2 (1+|\xi|)^{2\alpha} - c_3
\end{align}
with a strictly positive constant $c_2$ and $C_3', c_3\ge0$.

Moreover, piecewise continuity of $t\mapsto a_t(u,v)$ for every $u,v\in H^\alpha_\eta(\rrd)$ follows from the piecewise continuity of $t\mapsto A_t(z)$ for every $z\in U_{-\eta}$ and dominated convergence, which applies thanks to (Cont-$A$).

For the implication from $(i)$ to $(ii)$, we first show the following. Let $\gamma:\rrd\to\rr$ be a continuous function. 
If we have  
\begin{align}\label{help}
\int_{\rrd} \gamma(\xi) |\OF_{\eta'}(u)(\xi)|^2 \ee{-2\skl \eta',\xi\skr} \dd \xi \ge 0 
\end{align}
for all $u\in H^\alpha_{\eta'}(\rrd)$ for which $\OF_{\eta'}(u)$ is compactly supported,
then $\gamma(\xi)\ge0$ for all $\xi\in\rrd$. 
To prove this claim, we
follow closely the derivation of the fundamental lemma of variational calculus. 
Let us for a moment assume that $\gamma(\xi)<0$ for some $\xi\in\rrd$. Due to continuity, $\gamma$ is negative on a nonempty open subset of $U\subset\rrd$. We now choose a function $u \in H^{\alpha}_{\eta'}(\rrd)$ such that its Fourier transform $\OF_{\eta'}(u)$ is smooth, nonconstant and compactly supported in\tild $U$. For this choice of $u$, however, the integral in inequality~\eqref{help} would be negative, leading to a contradiction. This shows that $\gamma\ge0$. 


We observe that (Cont-$a$) implies inequality \eqref{help} for the continuous mappings $\xi\mapsto C(1+|\xi|)^{2\alpha} \pm \Re\big( A_t(\xi-i\eta')\big)$ and $\xi\mapsto C(1+|\xi|)^{2\alpha} \pm \Im\big( A_t(\xi-i\eta')\big)$ for all $t\in[0,T]$ and $\eta'\in R_{\eta}$. Thus (Cont-$A$) follows.
Similarly, using once again inequality \eqref{elementaryineq}, we obtain that (G{\aa}rd-$a$) implies (G{\aa}rd-$A$).

Finally, we observe that $\lim_{s\to t} a_s(u,u) = a_t(u,u)$ implies that
\begin{align*}
\lim_{s\to t} \int_{\rrd}  A_s(\xi - i\eta') \big|\OF_{\eta'} (u)(\xi)\big|^2 \dd \xi = \int_{\rrd}  A_t(\xi - i\eta') \big|\OF_{\eta'} (u)(\xi)\big|^2 \dd \xi,
\end{align*}
while on the other hand dominated convergence shows that
\begin{align*}
\lim_{s\to t} \int_{\rrd}  A_s(\xi - i\eta') \big|\OF_{\eta'} (u)(\xi)\big|^2 \dd \xi = \int_{\rrd}  \lim_{s\to t}A_s(\xi - i\eta') \big|\OF_{\eta'} (u)(\xi)\big|^2 \dd \xi.
\end{align*}
Now, an application of inequality \eqref{help} yields $\lim_{s\to t}A_s(\xi - i\eta') = A_t(\xi - i\eta')$ for all $\xi\in\rrd$ and $\eta'\in R_{\eta}$. Therefore piecewise continuity of the bilinear form entails piecewise continuity of the symbol.
\qed
\end{proof}

Theorem \ref{Theo-parabolicity} now follows as a Corollary from Theorem \ref{Prop-parabolicity}:
\begin{proof}[Theorem \ref{Theo-parabolicity}]
According to Lemma 2.1 (c) in \cite{EberleinGlau2013}, 
 for every $0\le t \le T$
the map $z\mapsto A_t(z)$ has a continuous extension to the domain $\overline{U_{-\eta}}$. 
Moreover, Theorem 25.17 in \cite{Sato} together with Lemma\tild \ref{lem-Aeta} in Appendix~\ref{sec-adjop} shows that inequality \eqref{somecont-A} is satisfied for each $t\in[0,T]$ with $m(t)=2$ and some constant $C(t)>0$. Theorem \ref{Theo-parabolicity} is now applicable and yields the corollary.
\qed
\end{proof}

\bibliographystyle{apalike}
  \bibliography{LiteraturFKac}
  
%

\end{document}